\documentclass[a4paper,11pt]{article}
\pdfoutput=1
\usepackage{jheppub}
\makeatletter
\gdef\@fpheader{}
\makeatother
\usepackage[a4paper,left=2.05cm,right=2.05cm,top=2.35cm,bottom=2.35cm]{geometry}
\usepackage{amsthm}
\usepackage{booktabs}
\usepackage{tabularx}
\usepackage{xcolor}
\usepackage{array}
\usepackage{mathtools}
\usepackage{longtable}
\usepackage{graphicx}
\usepackage{flafter}
\usepackage{microtype}
\usepackage{url}
\newcommand{\cO}{\mathcal O}
\newcommand{\PP}{\mathbb P}
\newcommand{\CC}{\mathbb C}
\newcommand{\dd}{\mathrm d}
\newcommand{\MA}{\mathrm{MA}}
\newcommand{\ESS}{\operatorname{ESS}}
\newcommand{\CVaR}{\operatorname{CVaR}}

\newtheorem{proposition}{Proposition}[section]
\newtheorem{corollary}[proposition]{Corollary}

\newtheorem{lemma}[proposition]{Lemma}

\title{\boldmath Positive Tensor-Network K\"ahler Metrics\\
on gCICY Threefolds}

\author{Juntao Wang}
\affiliation{Beijing Institute of Mathematical Sciences and Applications,
Beijing 101408, China}
\emailAdd{jtwang.bimsa@gmail.com}

\abstract{We develop and implement a positive tensor-network parameterization
for computing Ricci-flat K\"ahler metrics on Calabi--Yau manifolds.  It replaces
the large Hermitian coefficient matrix of a high-degree algebraic
metric by a matrix-product factorization.  For the immersed source spaces
used here, the resulting metric is globally positive for every parameter
value and, at fixed local and bond dimensions, its number of parameters
grows only linearly with the algebraic degree.  We test the construction on
three generalized complete-intersection Calabi--Yau (gCICY) threefolds,
constructing chart by chart the generalized sections, holomorphic volume
forms and sampling measures that define and train the metric there.  From a common low-degree metric, the tensor network outperforms a
parameter-matched neural potential using the same section data, reducing both
bulk errors and the one-percent tail conditional mean in every paired run.
It also reaches a substantially lower error than direct optimization of an
unrestricted Hermitian metric of the same degree from the same start, with
both methods optimized to validation convergence under their respective
schedules.  On a second geometry, a
higher-degree network with fewer parameters than a lower-degree unrestricted
Hermitian baseline substantially reduces the same-sample errors.  We further
observe saturation within the tested calculations: at fixed bond
dimension, increasing the degree eventually plateaus; increasing the bond
dimension at fixed optimization effort gives no resolved gain; and the outcome
depends strongly on initialization and optimization path.}


\begin{document}
\maketitle
\raggedbottom

\section{Introduction}
\label{sec:introduction}

Calabi--Yau threefolds are standard internal spaces in supersymmetric
string compactifications \cite{candelas1985vacuum}.  The massless spectrum
is often determined by topological and holomorphic data, whereas many
non-holomorphic observables are not: canonically normalized Yukawa
couplings, Kaluza--Klein masses and curvature-dependent corrections depend
on the Ricci-flat metric of the internal space
\cite{butbaia2024physical,berglund2025precision,braun2008spectrum,
cui2020curvature}.
Yau's theorem guarantees that a unique such metric exists in each K\"ahler
class \cite{yau1978}, but it does not
usually provide the metric in closed form.  Standard numerical approaches
include direct discretizations of the Ricci-flat equations on special
geometries; Donaldson's balanced metrics, which determine a positive Hermitian
form on a finite-dimensional space of holomorphic sections by a fixed-point
iteration; algebraic energy minimization, which varies such a Hermitian form
to reduce the variation of the Monge--Amp\`ere ratio; and neural K\"ahler
potentials, which use a neural network to parameterize a scalar potential or a
correction to a reference potential
\cite{headrick2005k3,donaldson2005,douglas2006,braun2008,
headrick2013energy,ashmore2019,anderson2021,jejjala2022neural,
berglund2023curvature,larfors2022}.  A recent pedagogical review compares
these principal approaches \cite{anderson2023lectures}.  Their outputs are
needed whenever an observable depends on the local metric rather than only on
topology, for example in tests of localized curvature scales relevant to
higher-derivative expansions \cite{cui2020curvature} and in Rayleigh--Ritz
approximations to the scalar Laplacian
\cite{braun2008spectrum,ashmore2023spectra}.  A recent roadmap surveys
the physical and mathematical questions that become accessible once such
metric data are available \cite{berglund2026roadmap}.
Current developments include reusable numerical software and moduli-dependent
neural solvers, symmetry-informed and globally defined K\"ahler-potential
models, and compact analytic or symbolic descriptions of metrics or
metric-derived observables
\cite{larfors2021size,gerdes2023cyjax,butbaia2025cymyc,hendi2025invariant,
mirjanic2025symbolic,lee2025approximate,eng2026symbolic,
mirjanic2026,rahman2026globalcy}.  Related work has constructed approximate
metrics with explicit K\"ahler-moduli dependence and, more recently,
simultaneous complex-structure and K\"ahler-moduli dependence
\cite{constantin2026kahler,constantin2026fullmoduli}.  Numerical Ricci-flat
metrics have also been combined with harmonic forms and a numerically solved
warp factor to construct warped Type IIB flux backgrounds
\cite{lust2026warped}.

For the algebraic approaches, accuracy is raised by raising the degree
\(k\), but on a polarized Calabi--Yau threefold the space of degree-\(k\)
sections of the polarizing line bundle \(L\) has dimension
\(N_k=O(k^3)\), so an unrestricted Hermitian
parameterization carries \(O(k^6)\) real
parameters.  Raising the degree therefore quickly becomes prohibitive.
Our first aim is a metric family whose parameter count grows
more slowly in \(k\) while remaining inside the positive algebraic class.  Tensor
networks are the natural candidate: for positive integers \(b\) and \(m\), repeated multiplication
\(H^0(X,L^b)^{\otimes m}\to H^0(X,L^{mb})\) presents high-degree sections
as products of low-degree ones, and a positive matrix-product
factorization of the Hermitian form in these coordinates has parameter
count linear in \(m=k/b\), with positivity built in.  The calculations
below test this compressibility hypothesis; they do not assume that the
bond dimensions required at fixed accuracy remain bounded asymptotically.

Our second aim is to compute metrics on generalized complete
intersections (gCICYs).  These enlarge the ordinary CICY construction by
allowing negative entries in later columns of the configuration matrix
\cite{anderson2016gcicy,berglund2016hirzebruch,garbagnati2017,jia2020},
contain members not realized by ordinary complete intersections, have also
been studied using computational and machine-learning methods
\cite{cui2023gcicy}, and have been used in
heterotic line-bundle model building \cite{larfors2020heterotic}.  Yet, to
our knowledge, no numerical Ricci-flat metric has been computed on any of
them: the
negative entry describes a line bundle on an intermediate complete
intersection, its sections are not polynomials on the ambient product,
and the defining equations, holomorphic volume form, sampling density,
section vector and tangent map must all be constructed in an atlas adapted
to the staged geometry and shown to transform consistently on overlaps.  The
two aims meet here: the tensor-network (TN) ansatz needs only a low-degree
section space and its multiplication, precisely the data that remain
computable on a gCICY, so the compressed family is particularly natural
on these spaces.

Three earlier constructions are close in motivation.  The
holomorphic networks of Ref.~\cite{douglas2022holomorphic} learn a
subspace of high-degree sections through homogeneous polynomial
activations and insert a Hermitian combination of them into an algebraic
K\"ahler potential.  The Grassmannian method of
Ref.~\cite{ek2024grassmannian} learns an efficient subspace of a fixed
section space and determines the Hermitian matrix on it by balanced
iteration or direct optimization.  The pseudoinverse iteration of
Ref.~\cite{berglund2026balancedsyz} adapts the balanced-metric fixed point
to ambient section data in a canonical monomial basis.  None factorizes
the lifted Hermitian operator in product coordinates formed by repeatedly
multiplying low-degree sections.  Ref.~\cite{douglas2022feedforward} had,
however, already proposed parameterizing Fubini--Study potentials by
matrix-product states, in a trace-form bihomogeneous sketch left
explicitly untried.  Our ansatz is structurally
different: rather than contracting the matrix-product state directly into
the bihomogeneous potential, we represent a factor of the lifted
Hermitian operator as a matrix-product operator, so that the coefficient
form is a manifest square plus a fixed positive reference block.  This
yields Hermiticity and positivity by construction, an exact degree
continuation, and---under an immersion condition---a globally positive
K\"ahler metric for every parameter value.

The tensor-network ingredients themselves---matrix-product states and
operators, and positivity through a factorized square \(B^\dagger
B\)---are standard
\cite{white1992dmrg,schollwock2011mps,murg2010mpo,oseledets2011tt,
werner2016positive}.  What is new is their geometric use: the
tensor-product coordinates are fixed by the multiplication map rather than
learned; the learned object is a matrix-product factor \(B_\theta\) in the
positive form \(B_\theta^\dagger B_\theta\), rather than a selected
subspace of sections; and the
contracted norm and its derivatives are evaluated without ever forming the
high-degree section vector or Hermitian matrix.  To our knowledge, this
combination has not previously been used to construct and train global
algebraic K\"ahler metrics \cite{douglas2022feedforward}; on gCICYs, the
stage-adapted atlas,
generalized sections, Poincar\'e residues and sampling machinery are
essential parts of the global realization.

We construct the stage-adapted generalized sections, Poincar\'e residues,
tangent maps and fibrewise sampling measures required to evaluate algebraic
metrics globally on gCICYs, and prove that the chosen TN source maps are
closed immersions.  Consequently, every parameter choice defines a global
K\"ahler metric in the fixed class, while the parameter count grows linearly
with the number of multiplication sites at fixed local and bond dimensions.

The three examples address complementary questions.  On \(X_{11}\), a
type-\((1,1)\) threefold, we compare the TN with a parameter-matched neural
K\"ahler-potential correction and with the unrestricted Hermitian family at
the same algebraic degree, all from a common initial metric.  On the K3-fibred
type-\((2,1)\) example \(X_{21}\), we ask whether compression makes a higher
algebraic degree useful at a fixed parameter scale and study separately the
effects of degree, bond dimension and initialization.  On the two-stage
type-\((2,2)\) geometry \(X_{22}\), we test the same construction with two
generalized defining sections and determine the effect of enlarging the local
operator space.  Together the examples show that additional formal capacity
improves the metric only when the optimization can exploit the new directions.
The TN is therefore a structured positive subfamily of the usual
finite-degree Hermitian ansatz: it makes selected high-degree families
computationally accessible, while unrestricted Hermitian forms remain
natural low-degree baselines and parent spaces for continuation.

We first formulate the Ricci-flat objective and the positive TN ansatz in
Sec.~\ref{sec:method}.  Section~\ref{sec:geometry} then constructs the
gCICY sections, residues, tangent data and sampling measures needed to
evaluate the metric globally.  With these ingredients in place,
Sec.~\ref{sec:results} presents the three numerical studies, and
Sec.~\ref{sec:conclusion} summarizes their implications and limitations
and outlines directions for future work.
The appendices collect the geometric derivations and numerical settings.

\section{Ricci-flat K\"ahler metrics and a positive tensor-network ansatz}
\label{sec:method}

\subsection{Ricci-flat K\"ahler metrics and algebraic approximations}
\label{sec:continuous-geometry}

Let \(X\) be a compact Calabi--Yau threefold.  We fix a K\"ahler
class \([\omega]\in H^{1,1}(X,\mathbb R)\) and choose a
nowhere-vanishing holomorphic \((3,0)\)-form
\(\Omega\).  Equivalently, \(\Omega\) is a nowhere-vanishing section
of the trivial canonical bundle \(K_X\).

Choose a holomorphic coordinate chart \(U\subset X\) with coordinates
\(u=(u^1,u^2,u^3)\).  Indices \(a,b\) range from \(1\) to \(3\), and
repeated holomorphic--antiholomorphic index pairs are summed.  We write
\(\partial_a=\partial/\partial u^a\) and
\(\bar\partial_{\bar b}=\partial/\partial\bar u^{\bar b}\).  Let \(g\) be a
K\"ahler metric whose associated K\"ahler form \(\omega_g\) represents the
fixed class \([\omega]\), and write \(g_{a\bar b}\) for its local
positive-definite Hermitian metric.  On \(U\) there is a real
function \(K_U\), called a local
K\"ahler potential, such that
\begin{equation}
 \omega_g=i\,g_{a\bar b}\,
 \dd u^a\wedge\dd\bar u^{\bar b}
 =i\,\partial\bar\partial K_U,
 \qquad
 g_{a\bar b}=\partial_a\bar\partial_{\bar b}K_U.
 \label{eq:local-kahler-potential}
\end{equation}
The potential \(K_U\) need not be a globally defined function.  If \(V\)
is a second coordinate chart, then on an overlap \(U\cap V\) there is a
holomorphic function \(f_{UV}\) for which the two local potentials may
differ by the K\"ahler transformation
\begin{equation}
 K_V=K_U+f_{UV}+\overline{f_{UV}},
 \label{eq:kahler-transformation}
\end{equation}
which leaves \(\partial\bar\partial K_U\), and hence the metric, unchanged.

The global holomorphic three-form also has a chart-dependent local
description, which defines its holomorphic coefficient \(h_U\) and the
positive density \(\rho_{\Omega,U}\):
\begin{equation}
 \left.\Omega\right|_U
 =h_U(u)\,\dd u^1\wedge\dd u^2\wedge\dd u^3,
 \qquad
 \rho_{\Omega,U}=|h_U(u)|^2.
 \label{eq:continuous-cy-data}
\end{equation}
Because \(\Omega\) never vanishes, \(h_U\) is
nowhere zero on \(U\).  The replacement \(\Omega\mapsto c\Omega\), with
\(c\in\mathbb C\setminus\{0\}\), is the standard normalization freedom of
a holomorphic volume form and is independent of our numerical method.  It
sends \(\rho_{\Omega,U}\mapsto |c|^2\rho_{\Omega,U}\).  Write
\(\det g_U=\det(g_{a\bar b})\) for the determinant of the \(3\times3\)
Hermitian metric on this chart.  Under a holomorphic change of coordinates
\(u\mapsto u'\),
\begin{equation}
 h_{U'}=h_U\det\!\left(\frac{\partial u}{\partial u'}\right),
 \qquad
 \det g_{U'}=
 \left|\det\!\left(\frac{\partial u}{\partial u'}\right)\right|^2
 \det g_U.
 \label{eq:metric-volume-coordinate-law}
\end{equation}
Thus neither \(\det g_U\) nor \(\rho_{\Omega,U}\) is separately a scalar,
but they transform by the same positive Jacobian factor.  Their ratio
therefore defines the global coordinate-invariant Monge--Amp\`ere function
\begin{equation}
 \eta_g=\frac{\det g_U}{\rho_{\Omega,U}}.
 \label{eq:continuous-ma-ratio}
\end{equation}
Equivalently, define the positive volume form associated with \(\Omega\) by
\begin{equation}
 \dd\mu_\Omega=i\,\Omega\wedge\bar\Omega.
 \label{eq:omega-volume-form}
\end{equation}
The phase is the convention appropriate to complex dimension three.
Then the local definition above is exactly the global volume-form identity
\begin{equation}
 \frac{\omega_g^3}{3!}=\eta_g\,\dd\mu_\Omega.
 \label{eq:ma-volume-identity}
\end{equation}

The symbol \(\operatorname{Ric}(\omega_g)\) denotes the Ricci
\((1,1)\)-form of the K\"ahler metric.  It obeys
\begin{equation}
 \operatorname{Ric}(\omega_g)
 =-i\,\partial\bar\partial\log\det g_U
 =-i\,\partial\bar\partial\log\eta_g,
 \label{eq:ricci-ma-relation}
\end{equation}
where the second equality uses
\(\partial\bar\partial\log|h_U|^2=0\).  Hence Ricci flatness is equivalent
to \(\eta_g\) being constant on compact connected \(X\).  If this constant
is denoted by \(\kappa\), Eq.~\eqref{eq:ma-volume-identity} fixes it as
\begin{equation}
 \kappa=
 \frac{\displaystyle\int_X\omega_g^3/3!}
      {\displaystyle\int_X\dd\mu_\Omega}.
 \label{eq:ma-target-constant}
\end{equation}
The numerator depends only on the fixed class \([\omega]\).  Equation
\eqref{eq:ma-volume-identity} therefore states the Ricci-flat condition
equivalently as proportionality between \(\omega_g^3\) and the positive
volume form \(\dd\mu_\Omega\).  Under \(\Omega\mapsto c\Omega\), both
\(\eta_g\) and \(\kappa\) are multiplied by \(|c|^{-2}\), so this equation
and its metric solution are unchanged.  Yau's theorem gives a unique
solution in the chosen class~\cite{yau1978}.

Numerically, the constant \(\kappa\) need not be supplied in advance.  The
normalized importance weights introduced in Sec.~\ref{sec:sampling}
approximate integration against
\(\dd\mu_\Omega/\int_X\dd\mu_\Omega\).  Hence
\(\sum_i w_i\eta_g(x_i)\), with \(\sum_iw_i=1\), is a Monte Carlo estimator
of \(\kappa\).  We divide \(\eta_g\) by this estimate and measure
the deviation of the resulting dimensionless ratio from one.  The overall
normalization of \(\Omega\) cancels in this quotient, whereas any spatial
variation in \(\eta_g\) remains visible.  We construct a tractable family
of local potentials obeying the transition law
\eqref{eq:kahler-transformation} and minimize these deviations.  The precise
sample-normalized ratio and the bulk and tail statistics are defined in
Sec.~\ref{sec:training}.

\paragraph{Unrestricted algebraic metric families.}

Let \(L\to X\) be an ample holomorphic line bundle, called a
polarization, and write \(c_1(L)\) for its first Chern class.  We use the
standard topological convention
\(c_1(L)=[\omega_{\rm FS}]/(2\pi)\), where
\(\omega_{\rm FS}=i\,\partial\bar\partial\log\lVert Z\rVert^2\) is the
Fubini--Study form in homogeneous coordinates \(Z\) on projective space.
Thus the
algebraic forms below represent the fixed class
\begin{equation}
 [\omega]=2\pi c_1(L).
 \label{eq:polarization-normalization}
\end{equation}
This is also the normalization used in the numerical calculations
and all dimensionful quantities reported below.  For a positive integer
\(k\), the notation \(L^k\) means
the \(k\)-fold tensor power of \(L\), and \(H^0(X,L^k)\) is the
finite-dimensional complex vector space of its global holomorphic sections.  Set
\(N_k=\dim_{\mathbb C}H^0(X,L^k)=h^0(X,L^k)\), choose an ordered basis
\((s_{k,1},\ldots,s_{k,N_k})\), and collect its values in the column vector
\(s_k(x)\).  Finally, let \(H\) be an \(N_k\times N_k\) positive-definite
Hermitian matrix and let \(\dagger\) denote Hermitian transpose.  In a
local frame of \(L^k\), the local K\"ahler potential of the corresponding
degree-\(k\) algebraic metric is
\begin{equation}
 K_H=\frac{1}{k}\log\!\left(s_k^\dagger Hs_k\right),
 \qquad H>0.
 \label{eq:full-H}
\end{equation}
Allowing all Hermitian entries of \(H\) to vary, subject only to \(H>0\)
and the irrelevant overall scale, gives the unrestricted degree-\(k\)
Hermitian family.
When the projective map defined by the complete space \(H^0(X,L^k)\)
is an embedding of \(X\), as verified for the
section spaces used below, the associated K\"ahler form is
\(\omega_H=i\,\partial\bar\partial K_H\).  It is global because a
line-bundle frame change adds only a K\"ahler transformation to \(K_H\).

An arbitrary matrix \(H\) does not solve the Ricci-flat equation.  Each
\(\omega_H\) lies in the class fixed in
Eq.~\eqref{eq:polarization-normalization}.  For the unique
Ricci-flat form \(\omega_{\rm RF}\) in this class, there is a sequence
\(H_k>0\) such that
\(\omega_{H_k}\to\omega_{\rm RF}\) smoothly
\cite{tian1990polarized,zelditch1998szego}.  For fixed \(k\), however,
varying \(H\) explores only a finite-dimensional subset of the K\"ahler
metrics in the chosen class, and the exact Ricci-flat metric need not lie
in this subset.  Since Ricci flatness is equivalent to the
Monge--Amp\`ere ratio \(\eta_H\) being constant, we choose \(H\) by
minimizing the spatial variation of \(\eta_H\).  A nonzero minimum
represents the approximation error of the finite-degree family.

Donaldson's balanced algorithm gives one prescription for choosing \(H\) in
this finite-dimensional family.  Starting from a positive matrix \(H_n\),
one first constructs the potential
\(K_{H_n}=k^{-1}\log(s_k^\dagger H_n s_k)\) and its associated K\"ahler
metric.  One then averages the normalized products of sections with
respect to a chosen smooth positive volume form \(\dd\mu\).  Writing
\(\operatorname{Vol}(X)=\int_X\dd\mu\), the update is
\begin{equation}
 T(H_n)_{\alpha\bar\beta}
 =\frac{N_k}{\operatorname{Vol}(X)}\int_X
 \frac{s_{k,\alpha}\overline{s_{k,\beta}}}
      {s_k^\dagger H_n s_k}\,\dd\mu,
 \qquad H_{n+1}\propto T(H_n)^{-1},
 \label{eq:balanced-t-map}
\end{equation}
where \(\propto\) denotes equality up to a positive scalar.  A balanced
metric is a fixed point up to this irrelevant overall scale, so that
\(T(H)\propto H^{-1}\).  Constructing \(K_{H_n}\) and then performing the
section-space integration are conventionally called the Fubini--Study and
Hilb steps, respectively; their composition gives the usual \(T\)-operator
\cite{donaldson2005,douglas2006,braun2008}.
Equation~\eqref{eq:balanced-t-map} is the balanced-metric iteration written
in the present \(H\)-convention.  For the
Calabi--Yau calculation, taking \(\dd\mu\) proportional to
\(\dd\mu_\Omega\) gives fixed points whose associated K\"ahler forms
converge to \(\omega_{\rm RF}\) as \(k\) increases
\cite{donaldson2005,douglas2006,braun2008}.
Energy-functional methods instead minimize a functional of the
Monge--Amp\`ere error directly.
Let \(g_H\) be the metric obtained from \(K_H\), let \(\eta_H\) denote its
Monge--Amp\`ere ratio from Eq.~\eqref{eq:continuous-ma-ratio}, and define
the measure average
\(\langle\eta_H\rangle=(\operatorname{Vol}(X))^{-1}
\int_X\eta_H\,\dd\mu\).  A representative squared Monge--Amp\`ere energy
over the same unrestricted Hermitian family is then
\begin{equation}
 E_{\rm MA}(H)=\int_X
 \left(\frac{\eta_H}{\langle\eta_H\rangle}-1\right)^2\dd\mu.
 \label{eq:full-ma-energy}
\end{equation}
Because \(\dd\mu\) is positive and has full support, this energy is zero
precisely when \(\eta_H\) is constant, which by
Eq.~\eqref{eq:ricci-ma-relation} is the Ricci-flat condition.  At finite
degree, its minimum over \(H\) measures the approximation error of the
chosen algebraic family; a fitted value above that minimum may additionally
contain optimization error.
Both balanced iteration and direct minimization of \(E_{\rm MA}\) can be
applied to the same unrestricted Hermitian ansatz \(K_H\)
\cite{headrick2013energy}.  They differ in how \(H\) is selected, not in the
number of available Hermitian coefficients.  Reference
\cite{berglund2026balancedsyz} instead reformulates the balanced iteration in
the canonical ambient monomial coordinates.  When the restriction map from
ambient sections to \(H^0(X,L^k)\) is surjective, its pseudoinverse
construction is equivalent to the standard iteration and produces the same
balanced metric on \(X\).  It preserves the ambient monomial labels but does
not reduce the intrinsic size of the Hermitian family.  The tensor-network
ansatz introduced below addresses a different problem: it replaces the
unrestricted matrix \(H\) by a structured factorization with fewer trainable
coefficients.

For a polarized threefold, Riemann--Roch gives \(N_k=O(k^3)\).  Denote the
raw number of real coefficients in an unrestricted Hermitian matrix
by \(P_{\rm Herm}(k)\).  Then
\begin{equation}
 P_{\rm Herm}(k)=N_k^2=O(k^6).
\label{eq:full-parameter-law}
\end{equation}
One overall positive rescaling of \(H\) adds a constant to \(K_H\) and
therefore leaves the metric unchanged; the parameter counts quoted below
include this single redundant scale.  Optimization over this family remains
highly effective at degrees for which storing and evaluating the Hermitian
matrix is computationally feasible; the purpose of the TN is to make selected
high-degree algebraic families accessible beyond that regime.

\subsection{Positive tensor-network ansatz}
\label{sec:tn-parameterization}

The object that we wish to parameterize is the same high-degree Hermitian
form that appears in Eq.~\eqref{eq:full-H}.  To emphasize its degree, denote
it here by \(H_k>0\).  Balanced iteration and energy minimization both use
increasing \(k\) as a refinement mechanism: the larger section space gives
the algebraic metric access to finer geometric structure.  The difficulty is
that an unrestricted \(H_k\) has \(N_k^2\) real coefficients.  Our TN is
therefore a structured ansatz for this high-degree Hermitian form itself.

Three requirements guide the ansatz.  First, a reliable metric found at an
affordable degree should be carried to a higher degree without changing its
K\"ahler potential; the additional high-degree freedom can then be trained
rather than used to reconstruct the starting geometry.  Second, every value
of the trainable parameters should define a global positive algebraic
K\"ahler metric.  Third, neither storage nor pointwise evaluation should
require the full high-degree Hermitian matrix.  These requirements lead respectively
to an exact degree lift, a factorized square, and a tensor-network
representation of that factor.

\paragraph{Exact degree continuation and section-ring realization.}

Choose an affordable base degree \(b\) and a positive integer \(m\), and set
\begin{equation}
 V=H^0(X,L^b),\qquad d=\dim_{\mathbb C}V,\qquad k=mb.
 \label{eq:degree-relation}
\end{equation}
Here \(m\) is the number of TN sites and \(k\) the degree of the
resulting algebraic metric.  Choose an ordered basis
\(s=(s_1,\ldots,s_d)^{\mathsf T}\) of \(V\) whose sections have no common
zero.  The superscript \({}^{\mathsf T}\) denotes ordinary transpose.  At
\(x\in X\), the vector \(s(x)\in\CC^d\) contains the section values in a
local frame of \(L^b\).

Before we give the index formulas, the construction can be pictured as three
maps.  Take \(m\) copies of the low-degree section vector, apply a trainable
linear map \(B_\theta\), and take the squared norm of the result:
\[
 s(x)\ \longmapsto\ v(x)=s(x)^{\otimes m}
 \ \xmapsto{\ B_\theta\ }\
 \Psi_\theta(x)
 \ \longmapsto\
 F_\theta(x)=\|\Psi_\theta(x)\|^2
 +\varepsilon_{\rm ref}
 \bigl(s(x)^\dagger H_0s(x)\bigr)^m .
\]
Here \(B_\theta\) is the trainable map, \(H_0>0\) is a fixed Hermitian form
on \(V\), and \(0<\varepsilon_{\rm ref}<1\) is fixed.  The first map raises the
section degree from \(b\) to \(k=mb\); the second introduces trainable
couplings among the degree-\(k\) product sections; and the squared norm,
together with the reference term, makes the resulting Hermitian form
positive.  The only remaining difficulty is that a general \(B_\theta\)
would itself be exponentially large.  Its tensor-network factorization is
what makes the ansatz computationally useful.

Suppose that \(H_b>0\) is a low-degree Hermitian form obtained, for example,
from balanced iteration or direct minimization of a geometric energy.  Write
\(F_b(x)=s(x)^\dagger H_b s(x)\), so that its local potential is
\(K_b=b^{-1}\log F_b\).  Supply one copy of \(s(x)\) to each of the \(m\)
sites and define the product vector
\begin{equation}
 v(x)=s(x)^{\otimes m},\qquad
 v_{i_1\ldots i_m}(x)=s_{i_1}(x)\cdots s_{i_m}(x).
 \label{eq:product-section-input}
\end{equation}
The tensor product gives an exact degree-\(k\) representation of the same
metric:
\begin{equation}
 F_k^{(0)}(x)=F_b(x)^m
 =v(x)^\dagger H_b^{\otimes m}v(x),
 \qquad
 \frac{1}{k}\log F_k^{(0)}(x)
 =\frac{1}{b}\log F_b(x).
 \label{eq:exact-degree-continuation}
\end{equation}
Thus the passage from \(b\) to \(k=mb\) initially changes only the
representation, not the K\"ahler potential or the metric.  In the
continuation experiments below, this identity supplies the initial
high-degree metric.  After initialization, the low-degree solution is no
longer used as a target: the new parameters are trained directly with the
Monge--Amp\`ere loss.

The exact lift has a bond-dimension-one operator factorization.  If
\(H_b=S_b^\dagger S_b\), then \(S_b^{\otimes m}\) is a product map and hence
has bond dimension one.  For the exact continuations used below, we take
\(H_0=H_b\) and initialize
\(B_\theta^{(0)}=\sqrt{1-\varepsilon_{\rm ref}}\,S_b^{\otimes m}\).  The
trainable and reference terms then sum exactly to \(H_b^{\otimes m}\).
Additional bond channels can then be inserted without changing the
represented function by setting the new block on one side to zero and
the block on the other side to a nonzero value, as specified in
Sec.~\ref{sec:metric-evaluation}; this enlarges the trainable family
without changing the starting metric.  This exactness statement applies to
the bond-dimension-one lifts of a base metric \(H_b\); re-embedding an already trained TN
at a higher degree is discussed in Sec.~\ref{sec:metric-evaluation}.

The ordered products in \(v(x)\) are degree-\(k\) sections, but they need
not be linearly independent.  Their relation to the conventional
degree-\(k\) section space is encoded by the multiplication map
\begin{equation}
 \mu_m:V^{\otimes m}\longrightarrow H^0(X,L^k),\qquad
 \mu_m(s_{i_1}\otimes\cdots\otimes s_{i_m})
 =s_{i_1}\cdots s_{i_m}.
 \label{eq:section-ring-multiplication-map}
\end{equation}
Permuting the factors does not change their product, and further relations
may follow from the defining equations of \(X\).  The image of \(\mu_m\)
is therefore the span of the distinct degree-\(k\) sections represented by
these products.  If \(\mu_m\) is surjective, this span is the complete space
\(H^0(X,L^k)\); otherwise the ansatz uses only the product-generated
subspace.  The map is used to check this coverage, not in pointwise
evaluation, where the factors \(s(x)\) are contracted directly.

\paragraph{Matrix-product representation and positive ansatz.}

The product representation gives access to higher degree, but it also
creates a computational problem: \(v(x)\) has \(d^m\) formal components,
and a general operator on these coordinates is exponentially large.  A
tensor network avoids forming that operator by contracting a sequence of
smaller tensors.  An open-chain matrix-product representation of an
order-\(m\) tensor \(\mathcal T\in(\CC^d)^{\otimes m}\), with
\(i_j=1,\ldots,d\), has the form
\begin{equation}
 \mathcal T_{i_1\ldots i_m}
 =\sum_{a_1,\ldots,a_{m-1}}
 G^{[1]}_{i_1a_1}
 G^{[2]}_{a_1i_2a_2}\cdots
 G^{[m]}_{a_{m-1}i_m}.
 \label{eq:generic-mps}
\end{equation}
The \(G^{[j]}\) are the site tensors.  The auxiliary index
\(a_j=1,\ldots,D_j\) is summed between two neighboring sites, and its range
\(D_j\) is the bond dimension.  This representation is called a
matrix-product state in many-body physics and a tensor train in numerical
analysis~\cite{schollwock2011mps,orus2014tensor,oseledets2011tt}.  Instead of
the \(d^m\) entries of a general tensor, the chain stores
\(\sum_j dD_{j-1}D_j\) entries, with \(D_0=D_m=1\).  Across the cut between
sites \(j\) and \(j+1\), the left--right matrix rank is at most \(D_j\).
The bond dimensions therefore control both the compression and the size of
the represented family.

Our model uses the operator analogue of Eq.~\eqref{eq:generic-mps}.  Each
site now has an input index \(i\), inherited from the corresponding factor
of \(v(x)\), and an output index \(p\).  Choose \(q\) matrices
\(Q_1,\ldots,Q_q\in\CC^{d\times d}\); these matrices form the local
operator dictionary.  At site \(j\), introduce the
trainable coefficient core \(C^{[j]}\) and the resulting site tensor
\(A^{[j]}\):
\begin{equation}
 C^{[j]}\in\CC^{D_{j-1}\times D_j\times q},
 \qquad
 A^{[j]}_{a_{j-1}a_jpi}
 =\sum_{\alpha=1}^q
 C^{[j]}_{a_{j-1}a_j\alpha}(Q_\alpha)_{pi}.
 \label{eq:local-map}
\end{equation}
The dictionary matrices may either be fixed or learned, depending on the
model.  When \(q=d^2\), a fixed matrix-unit dictionary already spans the
whole local operator space \(\operatorname{End}(\CC^d)\).  Learning another
basis of the same space would not enlarge the local linear span, because a
change of dictionary basis can be absorbed into the coefficients
\(C^{[j]}\); it would instead introduce an unnecessary parameter
redundancy.  We therefore fix \(Q_\alpha\) and train only the entries of all
\(C^{[j]}\) for \(X_{11}\) and \(X_{21}\), the geometries introduced in
Sec.~\ref{sec:geometry}.  For the truncated \(q<d^2\) dictionary used on
\(X_{22}\), the shared matrices \(Q_\alpha\)
are trained jointly with the coefficient cores, so that the selected local
operator subspace can adapt.  Accordingly, \(\theta\) denotes the
coefficients \(C^{[j]}\), together with the \(Q_\alpha\) whenever the
dictionary is learned.
The exact continuation described above belongs to this dictionary ansatz
provided
\begin{equation}
 S_b\in\operatorname{span}\{Q_1,\ldots,Q_q\}.
 \label{eq:dictionary-anchor-condition}
\end{equation}
This condition is automatic for the complete matrix-unit dictionary.  For
an exact continuation with a truncated learned dictionary, we impose it by
choosing one normalized dictionary element proportional to \(S_b\) and then
completing the remaining dictionary directions.
For fixed bond indices, \(A^{[j]}\) is a \(d\times d\) matrix from the
local input index \(i\) to the local output index \(p\).  Contracting the
bond indices defines
\begin{align}
 (B_\theta)_{p_1\ldots p_m,\,i_1\ldots i_m}
 &=\sum_{a_1,\ldots,a_{m-1}}
 \prod_{j=1}^m A^{[j]}_{a_{j-1}a_jp_ji_j},
 \qquad a_0=a_m=1,                                      \label{eq:generic-mpo}\\
 B_\theta&:(\CC^d)^{\otimes m}\longrightarrow
 (\CC^d)^{\otimes m}.                                  \nonumber
\end{align}
This is a matrix-product operator
\cite{murg2010mpo,schollwock2011mps,orus2014tensor}.  Applying it to the product section
vector gives
\begin{equation}
 \Psi_\theta(x)=B_\theta v(x),\qquad
 [\Psi_\theta(x)]_{p_1\ldots p_m}
 =\sum_{i_1,\ldots,i_m}
 (B_\theta)_{p_1\ldots p_m,\,i_1\ldots i_m}
 \prod_{j=1}^m s_{i_j}(x).
 \label{eq:explicit-contraction}
\end{equation}
This expression can be evaluated by contracting the chain successively;
neither the \(d^m\)-component vector nor the full operator \(B_\theta\) is
formed.

The final step converts this compressed linear map into a positive
Hermitian form.  For every \(v\in(\CC^d)^{\otimes m}\),
\begin{equation}
 v^\dagger B_\theta^\dagger B_\theta v
 =\|B_\theta v\|^2\geq0.
 \label{eq:factorized-positivity}
\end{equation}
The role of the fixed reference term introduced above is only to make the
form strictly positive; it is not a training target.  The product-coordinate
representative of our high-degree Hermitian ansatz is
\begin{equation}
 \widehat H_\theta=B_\theta^\dagger B_\theta+
 \varepsilon_{\rm ref}H_0^{\otimes m}>0,
 \label{eq:implicit-H}
\end{equation}
and its positive local homogeneous density is
\begin{equation}
 F_\theta(x)
 =v(x)^\dagger\widehat H_\theta v(x)
 =\|\Psi_\theta(x)\|^2+
 \varepsilon_{\rm ref}\bigl(s(x)^\dagger H_0s(x)\bigr)^m.
 \label{eq:positive-density}
\end{equation}
Equations~\eqref{eq:implicit-H} and \eqref{eq:positive-density} together
constitute the central ansatz of this work.  The first gives a compressed positive
representative in product coordinates, while the second evaluates the
corresponding degree-\(k\) algebraic metric.  Since every component of
\(v(x)\) is a section of \(L^k\), this is an ordinary Hermitian form on the
span of the product sections.  When \(\mu_m\) is surjective, that span is
the complete space \(H^0(X,L^k)\) used in Eq.~\eqref{eq:full-H}.
The factorization in Eq.~\eqref{eq:implicit-H} is related to positive
tensor-network constructions~\cite{werner2016positive}.  The resulting
local K\"ahler potential and metric components are
\begin{equation}
 K_\theta(x)=\frac{1}{k}\log F_\theta(x),\qquad
 g_{\theta,a\bar b}=\partial_a\bar\partial_{\bar b}K_\theta.
 \label{eq:potential}
\end{equation}
Their globality and strict metric positivity are established below.

Surjectivity of \(\mu_m\) guarantees that no degree-\(k\) section direction
is lost before the TN factorization is imposed; the bond dimensions then
determine which Hermitian forms in that section space the TN can represent.

The logic of the construction is thus: an affordable low-degree solver
provides a geometric starting point; Eq.~\eqref{eq:exact-degree-continuation}
embeds that metric exactly in higher-degree product coordinates; the TN
replaces the large high-degree Hermitian matrix by a positive
factorization; and direct Monge--Amp\`ere training uses the newly available
high-degree directions to improve on the starting metric.  The TN changes
the parameterization, not the Calabi--Yau equation or its geometric
objective.

\paragraph{Globality, K\"ahler class and positivity.}

We now explain why the fixed reference term in
Eq.~\eqref{eq:positive-density} guarantees a strictly positive metric, not
merely a positive local density \(F_\theta\).  A collection of
holomorphic section values defines a map from \(X\) to projective space.
The pullback of the Fubini--Study form by this map is positive definite
precisely when the map distinguishes every nonzero tangent direction, that
is, when it is an immersion.

Start with the low-degree basis \(s=(s_1,\ldots,s_d)^{\mathsf T}\) introduced
immediately after Eq.~\eqref{eq:degree-relation}.  Its values define the
projective source map
\begin{equation}
 \Phi_s:X\longrightarrow\PP^{d-1},
 \qquad
 \Phi_s(x)=[s_1(x):\cdots:s_d(x)].
 \label{eq:source-map-definition}
\end{equation}
The square brackets denote a projective point: two nonzero vectors represent
the same point when they differ by one common nonzero complex factor.  This
is why the local frame change \(s\mapsto fs\) does not change \(\Phi_s(x)\).
Two conditions on \(\Phi_s\) matter below, and they play different roles.
A point \(x\in X\) is called
a \emph{base point} of the chosen section space if every section in that
space vanishes at \(x\), or equivalently if
\(s_1(x)=\cdots=s_d(x)=0\).  The section space is \emph{basepoint free} when
it has no such point.  Thus at least one \(s_i(x)\) is nonzero at every
\(x\), which guarantees that the projective point \(\Phi_s(x)\) is defined
everywhere.  Once the map is defined, immersion is the additional
requirement that
\[
 \dd\Phi_s|_x:T_xX\longrightarrow T_{\Phi_s(x)}\PP^{d-1}
\]
be injective at every point.  Equivalently, no nonzero tangent vector on
\(X\) is collapsed by the map.  Basepoint freeness by itself does not imply
immersion.

Next consider the fixed reference term.  Since \(H_0>0\), it can be written
as \(H_0=S_0^\dagger S_0\) with \(S_0\) invertible, and
\[
 \bigl(s^\dagger H_0s\bigr)^m
 =\bigl\|S_0^{\otimes m}v(x)\bigr\|^2 .
\]
Thus the holomorphic feature vector contained in the reference term is
\(S_0^{\otimes m}v(x)\), and its associated projective map is
\begin{equation}
 \Phi_0:X\longrightarrow\PP^{d^m-1},\qquad
 \Phi_0(x)=\bigl[S_0^{\otimes m}v(x)\bigr]
 =\nu_m\!\left([S_0s(x)]\right),
 \label{eq:reference-projective-map}
\end{equation}
where we used \(v=s^{\otimes m}\), and where \(\nu_m\) is the tensor-power
form of the \(m\)-fold Veronese map.  The standard Veronese map lands in
\(\PP(\operatorname{Sym}^m\CC^d)\); the displayed target
\(\PP^{d^m-1}=\PP((\CC^d)^{\otimes m})\) uses redundant ordered tensor
coordinates and contains this symmetric subspace linearly.  Thus
\(\nu_m\) is still a projective embedding.  Equation
\eqref{eq:reference-projective-map} therefore has a simple interpretation:
\(\Phi_0\) is obtained from \(\Phi_s\) first by the invertible coordinate
change induced by \(S_0\), and then by the Veronese embedding.  Neither
operation collapses a tangent direction.  Consequently,
\[
 \Phi_s\text{ is an immersion}
 \quad\Longleftrightarrow\quad
 \Phi_0\text{ is an immersion}.
\]
The two maps therefore serve distinct roles.  Section~\ref{sec:geometry}
proves immersion for the geometrically natural low-degree map \(\Phi_s\),
whereas the positivity proof below uses the corresponding high-degree
reference map \(\Phi_0\).

\begin{proposition}[Globality, K\"ahler class and positivity]
\label{prop:positive}
Assume that the chosen section space \(V\) is basepoint free and that the source map
\(\Phi_s\) in Eq.~\eqref{eq:source-map-definition} is an immersion.  Then
Eqs.~\eqref{eq:positive-density} and
\eqref{eq:potential} define a global positive K\"ahler metric in the fixed
class \(2\pi c_1(L)\) for every value of the trainable parameters.
\end{proposition}

\begin{proof}
Under a local frame change \(s\mapsto fs\), where \(f\) is a
nowhere-vanishing holomorphic transition function, homogeneity gives
\(F_\theta\mapsto |f|^{2m}F_\theta\).  Thus \(K_\theta\) changes by the real
part of a holomorphic function, whose mixed derivative vanishes.  Hence
the mixed derivatives agree on chart overlaps, and
\(\omega_\theta=i\,\partial\bar\partial K_\theta\) defines a global real
\((1,1)\)-form.

More explicitly, introduce the direct-sum feature vector
\begin{equation}
 Z_\theta(x)=B_\theta v(x)\mathbin\oplus
 \sqrt{\varepsilon_{\rm ref}}\,S_0^{\otimes m}v(x).
 \label{eq:direct-sum-feature-map}
\end{equation}
Then \(F_\theta=\lVert Z_\theta\rVert^2\).  Its components are local
representatives of sections of \(L^k\), and the reference block never
vanishes.  They therefore define a projective map
\(\Xi_\theta:X\to\PP(\CC^{d^m}\oplus\CC^{d^m})\), with
\begin{equation}
 \omega_\theta=\frac{1}{k}\,\Xi_\theta^*\omega_{\rm FS}.
 \label{eq:fs-pullback}
\end{equation}
Here \(\omega_{\rm FS}=i\,\partial\bar\partial\log\lVert Z\rVert^2\) is
the standard positive Fubini--Study K\"ahler form on the target projective
space, and \(\Xi_\theta^*\omega_{\rm FS}\) denotes its pullback to \(X\).
Since \(\Xi_\theta^*\cO(1)\simeq L^k\) and
\([\omega_{\rm FS}]=2\pi c_1(\cO(1))\), the normalization by \(1/k\)
gives \([\omega_\theta]=2\pi c_1(L)\), as stated in
Eq.~\eqref{eq:polarization-normalization}.

It remains to prove strict positivity.  Let \(\pi\) be the projective
linear projection onto the reference summand.  Its center is the projective
subspace on which the reference block vanishes.  Since
\(S_0^{\otimes m}v(x)\ne0\) at every point, \(\Xi_\theta(X)\) does not meet
this center, so \(\pi\) is regular on the entire image and
\(\Phi_0=\pi\circ\Xi_\theta\).  If a tangent vector
\(\xi\in T_xX\) lay in the kernel of \(\dd\Xi_\theta|_x\), then
\[
 \dd\Phi_0|_x(\xi)
 =\dd\pi|_{\Xi_\theta(x)}\bigl(\dd\Xi_\theta|_x(\xi)\bigr)=0.
\]
The assumed immersion of \(\Phi_s\), together with
Eq.~\eqref{eq:reference-projective-map}, implies that \(\Phi_0\) is an
immersion and therefore forces \(\xi=0\).  Thus \(\Xi_\theta\) is an
immersion for every \(\theta\), and
the pullback in Eq.~\eqref{eq:fs-pullback} is positive definite.
\end{proof}

The immersion of the concrete source maps used in the three geometries is
proved in Sec.~\ref{sec:geometry}.

\paragraph{Parameter count and evaluation cost.}

For \(X_{11}\) and \(X_{21}\) we use the fixed complete matrix-unit
dictionary \(Q_{ab}=E_{ab}\), where \(E_{ab}\) has a single unit entry in
position \((a,b)\) and zeros elsewhere.  Thus \(q=d^2\), and we train only the
coefficient cores.  With \(m\geq2\) sites and uniform bond dimension
\(D\), their number of real trainable parameters is
\begin{equation}
 P_{\rm fixed}=2q\left[2D+(m-2)D^2\right].
 \label{eq:fixed-count}
\end{equation}
For the truncated dictionary used on \(X_{22}\), the \(q\) shared complex
\(d\times d\) matrices are also trained, adding \(2qd^2\) real parameters:
\begin{equation}
 P_{\rm learned}=2q\left[d^2+2D+(m-2)D^2\right].
 \label{eq:learned-count}
\end{equation}
These are raw trainable-parameter counts; virtual-bond gauge freedom and
dictionary-basis redundancy are not subtracted.  At fixed \((d,q,D)\) the
counts are linear in \(m\), but this does not imply that fixed ranks suffice
at arbitrary accuracy.

Pointwise metric evaluation is linear in \(m\) for the same reason.  Using
the holomorphic direct-sum feature vector \(Z_\theta\) of
Eq.~\eqref{eq:direct-sum-feature-map}, including the reference branch, one
has \(F_\theta=Z_\theta^\dagger Z_\theta\) and
\begin{equation}
 \partial_a F_\theta=Z_\theta^\dagger\partial_a Z_\theta,
 \qquad
 \partial_a\bar\partial_{\bar b}F_\theta
 =(\partial_b Z_\theta)^\dagger\partial_a Z_\theta.
 \label{eq:first-jets-suffice}
\end{equation}
Thus only first derivatives of the holomorphic section features are needed.
Sequential contraction propagates \(F_\theta\), its first derivatives and
its mixed derivatives without forming either the complete degree-\(mb\)
section vector or the full Hermitian matrix.  Hence storage and pointwise
evaluation are \(O(m)\) at fixed local and bond dimensions.  The numerical
normalization is described in Appendix~\ref{app:pointwise-evaluation}.

Table~\ref{tab:parameter-comparison} applies this counting law to the
three geometries studied in Sec.~\ref{sec:results}, alongside the
unrestricted Hermitian matrices of the same degree.
\begin{table}[t]
\centering
% Generated deterministically from geometry and architecture data by
% scripts/generate_gcicy_tn_final_tables.py.
\begin{tabular}{lrrrrr}
\toprule
type & target degree & \(h^0(L^k)\) & \(P_{\rm TN}\) & \(P_{\rm Herm}\) & \(P_{\rm Herm}/P_{\rm TN}\) \\
\midrule
\((1,1)\) & 6 & 871 & 141,750 & 758,641 & 5.4 \\
\((2,1)\), \(D=8\) & 8 & 2,440 & 96,800 & 5,953,600 & 61.5 \\
\((2,1)\), \(D=12\) & 8 & 2,440 & 214,896 & 5,953,600 & 27.7 \\
\((2,1)\), \(D=10\) & 16 & 19,216 & 343,640 & 369,254,656 & 1,074.5 \\
\((2,1)\), \(D=12\) & 16 & 19,216 & 493,680 & 369,254,656 & 748.0 \\
\((2,2)\) & 6 & 1,852 & 47,880 & 3,429,904 & 71.6 \\
\bottomrule
\end{tabular}

\caption{Structured and unrestricted parameter counts at selected target
degrees.  The columns \(P_{\rm TN}\) and \(P_{\rm Herm}\) give the raw real
coefficient counts of the tensor network and the corresponding unrestricted
Hermitian matrix.  The table compares representation size; accuracy is
studied separately in Sec.~\ref{sec:results}.}
\label{tab:parameter-comparison}
\end{table}
The gap widens quickly with degree: \(h^0(X_{21},L^{16})=19{,}216\), so an
unrestricted Hermitian matrix at that degree would carry
\(369{,}254{,}656\) real parameters, whereas the tested TNs carry
approximately \((3.4\text{--}4.9)\times10^5\).  Thus the factorization opens
a high-degree regime in which the unrestricted coefficient matrix is no
longer practical; it does not alter the inclusion of the TN family in the
corresponding unrestricted Hermitian family.

A complex128 benchmark over \(m=3,\ldots,8\) confirms the linear pointwise
scaling, with coefficient of determination \(R^2=0.9999\); timing and
normalization details are given in Appendix~\ref{app:pointwise-evaluation}.

\section{gCICY geometry and metric evaluation}
\label{sec:geometry}

\subsection{gCICY geometry, generalized sections, and source spaces}

This section explains how gCICY geometry supplies the global
section data, tangent maps, the Calabi--Yau density \(\rho_\Omega\) obtained
from the Poincar\'e residue, and the integration weights required
by the TN parameterization above.  Let
\(A=\prod_{\alpha=1}^{N_A}\PP^{n_\alpha}\), where \(N_A\) is the number of
ambient projective factors.  An ordinary CICY of codimension \(c\) is
described by
\begin{equation}
 X_{\mathrm{CICY}}=
 \left[
 \begin{array}{c|ccc}
  \PP^{n_1}&d_1^1&\cdots&d_c^1\\
  \vdots&\vdots&&\vdots\\
  \PP^{n_{N_A}}&d_1^{N_A}&\cdots&d_c^{N_A}
 \end{array}\right],
 \qquad d_a^\alpha\in\mathbb Z_{\geq0}.
 \label{eq:general-cicy-configuration}
\end{equation}
The \(a\)-th column is the multidegree
\(\mathbf d_a=(d_a^1,\ldots,d_a^{N_A})\) of a polynomial
\(P_a\in H^0(A,\cO_A(\mathbf d_a))\), and the simultaneous zero locus of
these polynomials is the complete intersection \(X_{\mathrm{CICY}}\).

Assume that the defining sections form a regular sequence and that their
common zero locus is smooth of the expected codimension.  Its dimension is
then the first quantity below.  If the second, Calabi--Yau row-sum condition
also holds, adjunction gives a trivial canonical bundle:
\begin{equation}
 \dim_\CC X=\sum_{\alpha=1}^{N_A}n_\alpha-c,
 \qquad
 \sum_{a=1}^{c}d_a^\alpha=n_\alpha+1
 \quad (\alpha=1,\ldots,N_A).
 \label{eq:gcicy-basic-conditions}
\end{equation}

A gCICY generalizes this construction by imposing the equations in stages
\cite{anderson2016gcicy}.  In the two-stage case one first constructs
\begin{equation}
 M=
 \left[
 \begin{array}{c|ccc}
  \PP^{n_1}&a_1^1&\cdots&a_{r_A}^1\\
  \vdots&\vdots&&\vdots\\
  \PP^{n_{N_A}}&a_1^{N_A}&\cdots&a_{r_A}^{N_A}
 \end{array}\right]
 \subset A,
 \qquad a_i^\alpha\in\mathbb Z_{\geq0},
 \label{eq:general-intermediate-configuration}
\end{equation}
and then defines
\begin{equation}
 X=
 \left[
 \begin{array}{c|ccc|ccc}
  \PP^{n_1}&a_1^1&\cdots&a_{r_A}^1&b_1^1&\cdots&b_{r_B}^1\\
  \vdots&\vdots&&\vdots&\vdots&&\vdots\\
  \PP^{n_{N_A}}&a_1^{N_A}&\cdots&a_{r_A}^{N_A}
  &b_1^{N_A}&\cdots&b_{r_B}^{N_A}
 \end{array}\right]
 \subset M,
 \qquad b_j^\alpha\in\mathbb Z.
 \label{eq:general-gcicy-configuration}
\end{equation}
Here \(r_A\) and \(r_B\) are the numbers of first- and second-stage defining
sections.  The second block is imposed by defining sections
\(q_j\in H^0(M,\cO_M(b_j^1,\ldots,b_j^{N_A}))\).  The total codimension is
\(c=r_A+r_B\), and Eq.~\eqref{eq:gcicy-basic-conditions} applies to the full
list of first- and second-stage columns.
Writing \(\mathbf a_i=(a_i^1,\ldots,a_i^{N_A})\) and
\(\mathbf b_j=(b_j^1,\ldots,b_j^{N_A})\) for their multidegrees, we identify
each multidegree with the corresponding line-bundle class on \(A\).  Let
\(K_A\), \(K_M\), and \(K_X\) denote the canonical bundles of \(A\), \(M\),
and \(X\), respectively, and use additive notation for tensor products of
line bundles.  Adjunction applied stage by stage then gives
\begin{equation}
 K_M=\left(K_A+\sum_{i=1}^{r_A}\mathbf a_i\right)\!\big|_M,
 \qquad
 K_X=\left(K_A+\sum_{i=1}^{r_A}\mathbf a_i
                 +\sum_{j=1}^{r_B}\mathbf b_j\right)\!\big|_X.
 \label{eq:sequential-adjunction}
\end{equation}
Thus the row-sum condition for the complete list of columns trivializes
\(K_X\) also in the staged construction.
The essential difference is now visible in the matrices.  For a CICY, every
column has non-negative degree and comes from a polynomial on \(A\).  For a
gCICY, a later column may contain \(b_j^\alpha<0\): it is not an ambient
polynomial, but it can still define \(X\) when the corresponding line bundle
has a global section on \(M\).  Additional stages are obtained by repeating
the same construction.

Two kinds of sections enter the calculation and should not be confused.
The \emph{defining sections}, denoted below by \(p_i\) and \(q_j\), cut out
the variety; their Jacobian determines the tangent space and the
Poincar\'e residue.  The \emph{source sections}
\(s=(s_1,\ldots,s_d)^{\mathsf T}\) belong to a positive line bundle on the
already constructed threefold.  They impose no further equations and
instead provide the projective features used by the TN in
Sec.~\ref{sec:method}.

This separation is why the gCICY calculation is not a routine reuse of a
CICY or hypersurface sampler.  Negative columns require regular local
representatives on intermediate varieties, equation transitions become
lower triangular rather than purely diagonal, fibrewise root
multiplicities depend on the stage structure, and the local representatives
of \(\Omega\) and the source jets must be transported through the same
atlas.  Once these data have been computed, the same TN contraction can be applied without reference
to the particular form of the defining equations.

In this paper, we test the same metric construction on the following three
fixed gCICY threefolds.  The model \(X_{11}\) is the Hirzebruch-family member
with index \(m_{\rm H}=3\) from
Ref.~\cite{berglund2016hirzebruch}, while \(X_{22}\) is the type-\((2,2)\)
example of Ref.~\cite{anderson2016gcicy}:
\begin{align}
 X_{11}&=
 \left[
 \begin{array}{c|c|c}
  \PP^4&1&4\\
  \PP^1&3&-1
 \end{array}\right],
 &
 X_{21}&=
 \left[
 \begin{array}{c|cc|c}
  \PP^5&1&2&3\\
  \PP^1&1&2&-1
 \end{array}\right],
 \label{eq:x11-x21}\\
 X_{22}&=
 \left[
 \begin{array}{c|cc|cc}
  \PP^1&1&1&-1&1\\
  \PP^1&1&1&1&-1\\
  \PP^5&3&1&1&1
 \end{array}\right].
 \label{eq:x22}
\end{align}
The fixed integral coefficient arrays used below define connected smooth
threefolds, as verified by exact finite-field smoothness and
Koszul-cohomology calculations~\cite{reproarchive}.

\paragraph{Chartwise regular representatives and covariance of the Poincar\'e residue.}
\label{sec:representatives}

The negative entries in Eqs.~\eqref{eq:x11-x21}--\eqref{eq:x22} do not
denote polynomials on the original ambient product.  They denote sections
of line bundles on an intermediate complete intersection and must therefore
be represented separately in local line-bundle frames.  The local
frame plays the role of a basis vector for the one-dimensional fibre of the
line bundle.  After a frame is chosen, a section is represented by an
ordinary function; changing the frame multiplies that function by a
nowhere-zero transition factor.  For a gCICY there is one
additional freedom: a representative of a later equation may be changed by
adding an earlier equation, because that earlier equation already vanishes
on the intermediate variety.  Thus two local representatives need not agree
as ambient functions.  They must agree after restriction to the appropriate
intermediate variety, together with the first-order data used to define the
tangent space and the holomorphic volume form.

\paragraph{First-order equations and the conormal sheaf.}
On an affine ambient chart, let \(I_X\) be the ideal sheaf of \(X\).  The
quotient \(I_X/I_X^2\simeq N^*_{X/A}\) records the first-order parts of the
defining equations.  For the smooth embeddings considered here it is the
conormal bundle and its local bases are supplied by any ordered regular
sequence defining \(X\).

Now let \(U,V\) be overlapping ambient charts and let
\[
 F_U=(f_{U,1},\ldots,f_{U,c})^{\mathsf T},\qquad
 F_V=(f_{V,1},\ldots,f_{V,c})^{\mathsf T}
\]
be ordered local defining equations for the same embedding
\(X\hookrightarrow A\).  They are regular sequences; in the present smooth
setting this means that their differentials are independent along \(X\).
Consequently the classes \([f_{U,1}],\ldots,[f_{U,c}]\) and
\([f_{V,1}],\ldots,[f_{V,c}]\) are two local frames of the same rank-\(c\)
locally free sheaf \(I_X/I_X^2\simeq N^*_{X/A}\).  They are related by a
matrix of regular functions whose determinant is nonzero on the overlap, so
\begin{equation}
 [F_V]=A_{VU}[F_U]\quad\text{in }I_X/I_X^2.
 \label{eq:conormal-change}
\end{equation}
Choosing the local defining equations in an order adapted to the stage
filtration makes \(A_{VU}\) block lower triangular.  The equations defining
an earlier intermediate variety are already local frames of its ideal and
cannot acquire a component from an equation imposed only at a later stage.
A later representative may, conversely, be changed by an earlier equation.
For one equation at each stage this gives
\((p_V,q_V)=(a\,p_U,b\,p_U+\gamma\,q_U)\) and hence
\(A_{VU}=\left(\begin{smallmatrix}a&0\\ b&\gamma\end{smallmatrix}\right)\);
the same statement holds blockwise.  Its diagonal blocks are the
nonvanishing changes of line-bundle frame, while an off-diagonal block
records the freedom to add an earlier defining equation to a later
generalized representative.  Equivalently,
\(F_V=A_{VU}F_U+R\) with \(R\in I_X^2\): the lower-triangular freedom
changes the displayed equations, but not the conormal space or the
intrinsic tangent space.
Lemma~\ref{lem:residue-covariance} in Appendix~\ref{app:glue} turns this
statement into the transformation law of the Poincar\'e residue,
\(\Omega_V=J_{VU}(\det A_{VU})^{-1}\Omega_U\) with
\(J_{VU}=\det(\partial z_V/\partial z_U)\).  By
Eq.~\eqref{eq:sequential-adjunction} this combined factor is precisely
the transition of \(K_X\), which the row-sum condition trivializes.  The
local Poincar\'e residues are therefore coordinate representatives of one
global nowhere-zero three-form \(\Omega\); the chart-by-chart verification for
the three geometries is recorded in Appendix~\ref{app:glue}.
Writing the global three-form in intrinsic coordinates as
\(\Omega=h_u\,\dd u^1\wedge\dd u^2\wedge\dd u^3\), an intrinsic recharting
\(u\mapsto u'\) gives
\begin{equation}
 h_{u'}=h_u\det\!\left(\frac{\partial u}{\partial u'}\right),
 \qquad
 \rho_{\Omega,u'}=\rho_{\Omega,u}
 \left|\det\!\left(\frac{\partial u}{\partial u'}\right)\right|^2,
 \label{eq:residue-density-change}
\end{equation}
where \(\rho_\Omega=|h_u|^2\).  After the projective-frame transition has
been accounted for as above, changing intrinsic coordinates is the only
remaining transformation.  The Hermitian metric determinant has the same
squared intrinsic Jacobian factor as \(\rho_\Omega\).  Therefore
\(\eta=\det g/\rho_\Omega\), the coordinate-invariant Monge--Amp\`ere
function introduced in Eq.~\eqref{eq:continuous-ma-ratio}, is independent
of both the projective chart and the nonsingular choice of implicit
coordinates.  The same line-bundle transition factors enter the proposal
density and cancel consistently in the normalized importance weights.  The sample-normalized ratio
\(r\), rather than \(\eta\) itself, is defined in Sec.~\ref{sec:training}.

On the three geometries this construction is carried out explicitly:
each generalized equation is realized by a pair of chartwise regular
representatives whose overlap difference is a combination of the earlier
equations---exactly the lower-triangular freedom of
Eq.~\eqref{eq:conormal-change}.  The explicit representatives and their
algebraic or cohomological derivations are given in
Appendices~\ref{app:h11}--\ref{app:h22}.

\paragraph{Source spaces and immersion.}
\label{sec:source-immersion}

For each geometry, the metric construction uses a basis
\(s=(s_1,\ldots,s_d)^{\mathsf T}\) of a low-degree section space
\(V_b=H^0(X,L^b)\).  The choices used in the numerical calculations are
\begin{center}
\begin{tabular}{c|c|c|c|c}
geometry & polarization \(L\) & source multidegree & \(b\) &
\(d=h^0(X,L^b)\)\\
\hline
\(X_{11}\) & \(\cO_X(1,1)\)   & \((2,2)\)     & \(2\) & \(45\)\\
\(X_{21}\) & \(\cO_X(1,1)\)   & \((1,1)\)     & \(1\) & \(11\)\\
\(X_{22}\) & \(\cO_X(1,1,1)\) & \((1,1,1)\) & \(1\) & \(17\)
\end{tabular}
\end{center}
For \(X_{21}\) and \(X_{22}\), the degree-one ambient monomials already
span the required restricted section spaces.  For \(X_{11}\), their
degree-one restriction has rank \(10<h^0(X_{11},L)=11\) and misses one
generalized section.  We therefore use the first complete ambient
restriction, at degree two, for which \(h^0(X_{11},L^2)=45\).  This
degree is chosen to obtain a complete source representation; the
degree-one subsystem would already suffice for immersion.

Threefold Riemann--Roch and Kodaira vanishing give
\begin{align}
 h^0(X_{11},L^k)&=\frac{23k^3+43k}{6},&
 h^0(X_{21},L^k)&=\frac{28k^3+38k}{6},&
 h^0(X_{22},L^k)&=\frac{50k^3+52k}{6}.
 \label{eq:rr}
\end{align}
The intersection-theoretic derivation is given in
Appendix~\ref{app:rr}.  These formulas determine both the source
dimensions above and the sizes of the corresponding unrestricted Hermitian
parameter spaces.

\begin{proposition}[Immersion of the source maps]
\label{prop:headline-immersion}
The source vectors used for \(X_{11}\), \(X_{21}\), and \(X_{22}\)
define closed immersions
\[
 \Phi_s:X\longrightarrow\PP^{d-1},
 \qquad
 x\longmapsto[s_1(x):\cdots:s_d(x)].
\]
\end{proposition}

\begin{proof}
The exact finite-field minors recorded in
Appendix~\ref{app:immersion} prove that the implemented source vectors
have ranks \(45\), \(11\), and \(17\).  Equation~\eqref{eq:rr} gives
the matching dimensions of the complete restricted section spaces, so
these vectors are bases rather than proper subsets.

The corresponding ambient line bundles are
\[
 \cO_{\PP^4\times\PP^1}(2,2),\qquad
 \cO_{\PP^5\times\PP^1}(1,1),\qquad
 \cO_{\PP^1\times\PP^1\times\PP^5}(1,1,1).
\]
They are very ample and define Segre--Veronese or Segre closed
embeddings of the ambient products.  Restricting these embeddings to
\(X\) preserves a closed immersion.  Removing linear relations that
vanish identically on \(X\) merely replaces the projective span of the
image by an isomorphic projective coordinate space.  The source maps
defined by the chosen bases are therefore closed immersions.
\end{proof}

\begin{corollary}
\label{cor:positive-families}
For every reported \(X_{11}\), \(X_{21}\), and \(X_{22}\) TN family,
Proposition~\ref{prop:positive} applies: the resulting K\"ahler metric
is globally positive for every value of the trainable parameters.
\end{corollary}

\subsection{Sampling and pointwise metric evaluation}
\label{sec:sampling}

The numerical loss and error measures require integrals with respect to
the Calabi--Yau volume form
\[
 \dd\mu_\Omega=i\,\Omega\wedge\bar\Omega.
\]
Directly generating independent points from this measure is not convenient.
We instead construct points using the fibrations of the three gCICYs and
then correct the resulting sampling measure by importance weights.  This
adapts the random-linear-section sampling used for projective
Calabi--Yau hypersurfaces and complete intersections in
Refs.~\cite{douglas2006,braun2008} to the fibre structure of the three
gCICYs.

\paragraph{Complete fibre intersections.}

Let
\[
 \pi:X\longrightarrow B
\]
denote the relevant fibration.  A single sampling step begins by drawing
a point \(\beta\in B\) from the Fubini--Study measure on the projective base.
Substituting the coordinates of \(\beta\) into the defining equations leaves
a positive-dimensional projective fibre
\[
 X_\beta=\pi^{-1}(\beta).
\]
We then choose a Haar-random complementary projective subspace
\(\Lambda\) in the projective space containing the fibre.  Its dimension
is chosen so that \(X_\beta\cap\Lambda\) is zero dimensional.  All isolated
points of this transverse intersection are included together as one
sampling cluster.

The construction is explicit in the three geometries:
\begin{center}
\begin{tabular}{c|c|c|c|c}
geometry & base \(B\) & fibre \(X_\beta\) &
random subspace \(\Lambda\) & points per cluster\\
\hline
\(X_{11}\) &
\(\PP^1\) &
quartic surface in \(\PP^3\) &
\(\PP^1\) &
\(4\)\\
\(X_{21}\) &
\(\PP^1\) &
\((2,3)\) surface in \(\PP^4\) &
\(\PP^2\) &
\(6\)\\
\(X_{22}\) &
\(\PP^1\times\PP^1\) &
plane cubic in \(\PP^2\) &
\(\PP^1\) &
\(3\)
\end{tabular}
\end{center}
For \(X_{11}\), fixing the \(\PP^1\) coordinate leaves a linear equation
in \(\PP^4\), hence a \(\PP^3\), together with a quartic equation on that
\(\PP^3\).  A generic projective line meets the resulting quartic surface
in four points.  For \(X_{21}\), fixing the base coordinate leaves a
quadratic and a cubic equation in \(\PP^4\); a generic projective plane
meets this complete-intersection surface in \(2\cdot3=6\) points.  For
\(X_{22}\), fixing the two \(\PP^1\) coordinates leaves three linear
conditions in \(\PP^5\), which determine a \(\PP^2\), and one cubic
equation on that plane.  A generic line therefore gives three points.
Thus the cluster sizes are
\begin{equation}
 c_{11}=4,\qquad c_{21}=6,\qquad c_{22}=3.
 \label{eq:cluster-sizes}
\end{equation}

The complex-projective kinematic, or Crofton, formula
\cite{howard1993} implies that averaging the complete intersections
\(X_\beta\cap\Lambda\) over Haar-random \(\Lambda\) produces the
Fubini--Study measure on the fibre.  It is important here to include all
intersection points: the Crofton identity applies to the sum over the
complete intersection, not to a root selected by a numerical preference.
Together with the independent Fubini--Study sampling on the base, this
defines a proposal measure \(\dd\nu_{\rm prop}\) on \(X\).  Up to
family-dependent constants, which cancel after normalization, its volume
forms can be written as follows.  Label the ambient projective factors by
their homogeneous-coordinate groups \(\alpha=x,y,z\), with \(z\) absent in
the two-factor examples.  Let
\(\pi_\alpha:A\to\PP^{n_\alpha}\) be projection onto the corresponding
factor and let \(\iota:X\hookrightarrow A\) be the inclusion.  If
\(\omega_{\rm FS}^{(\alpha)}\) is the Fubini--Study form on that factor, set
\[
 \omega_\alpha=(\pi_\alpha\circ\iota)^*
 \omega_{\rm FS}^{(\alpha)} .
\]
Then
\begin{equation}
 \dd\nu_{\rm prop}\ \propto\
 \begin{cases}
  \omega_y\wedge\omega_x^2/2,
    &X_{11},X_{21},\\[1mm]
  \omega_x\wedge\omega_y\wedge\omega_z,
    &X_{22},
 \end{cases}
 \label{eq:proposal-forms}
\end{equation}
The Crofton derivation and the fibre intersection classes are given in
Appendix~\ref{app:crofton}.

\paragraph{Conversion to the Calabi--Yau measure.}

The proposal measure generated above is generally not the Calabi--Yau
measure.  In local intrinsic coordinates
\(u=(u^1,u^2,u^3)\), write the two volume forms using the same coordinate
volume convention:
\[
 \dd\nu_{\rm prop}
 =\rho_{{\rm prop},u}\,\dd\lambda_u,
 \qquad
 \dd\mu_\Omega
 =\rho_{\Omega,u}\,\dd\lambda_u.
\]
Since both measures are expressed relative to the same local coordinate
volume element, their density ratio is
\[
 \frac{\dd\mu_\Omega}{\dd\nu_{\rm prop}}
 =
 \frac{\rho_{\Omega,u}}{\rho_{{\rm prop},u}}.
\]
For sampled points \(x_i\), let \(u^{(i)}\) denote the intrinsic coordinate
system used at \(x_i\).  We use the self-normalized importance weights
\begin{equation}
 w_i=
 \frac{\rho_{\Omega,u^{(i)}}(x_i)/\rho_{{\rm prop},u^{(i)}}(x_i)}
 {\displaystyle
  \sum_j\rho_{\Omega,u^{(j)}}(x_j)/\rho_{{\rm prop},u^{(j)}}(x_j)},
 \qquad
 \sum_iw_i=1.
 \label{eq:importance-weight}
\end{equation}
Under a change of intrinsic coordinates, both local densities acquire the
same squared Jacobian factor.  Their ratio, and hence \(w_i\), is therefore
coordinate independent.  The unknown overall normalization of the
Crofton measure also cancels in Eq.~\eqref{eq:importance-weight}, as does
the normalization freedom \(\Omega\mapsto c\Omega\).

For any integrable function \(f\), the weighted estimator satisfies
\begin{equation}
 \frac{\displaystyle\int_X f\,\dd\mu_\Omega}
      {\displaystyle\int_X\dd\mu_\Omega}
 \simeq
 \sum_iw_i f(x_i).
 \label{eq:importance-mean}
\end{equation}
Only if the proposal measure already equals the normalized
Calabi--Yau measure do the weights reduce to \(w_i=1/N_{\rm pt}\), where
\(N_{\rm pt}\) is the number of sampled points.

The ratio \(\rho_\Omega/\rho_{\rm prop}\) in
Eq.~\eqref{eq:importance-weight} corrects the sampling distribution.  It
should not be confused with the Monge--Amp\`ere ratio
\[
 \eta_g=\frac{\det g}{\rho_\Omega},
\]
which compares the volume form of a candidate metric with the
Calabi--Yau volume form.  Combining the two ratios gives
\begin{equation}
 \sum_iw_i\eta_g(x_i)
 \simeq
 \frac{\displaystyle\int_X\omega_g^3/3!}
      {\displaystyle\int_X\dd\mu_\Omega}
 =\kappa.
 \label{eq:kappa-importance-estimate}
\end{equation}
The sample-normalized Monge--Amp\`ere residual used for training and
evaluation is defined in Sec.~\ref{sec:training}.

\paragraph{Acceptance and statistical dependence.}

A fibre intersection contributes to the sample only when all expected roots
satisfy the defining equations and admit nonsingular intrinsic coordinates;
all of its roots are then included.  The importance weights and
Monge--Amp\`ere residuals play no role in this decision.  Roots from the same
intersection share a base point and a random projective subspace, so
uncertainty estimates resample complete intersections rather than treating
their roots as independent.  Numerical tolerances and sample diagnostics are
given in Appendix~\ref{app:crofton}.

\paragraph{Pointwise tangent, section-jet and Poincar\'e-residue computation.}
We spell out the calculation performed after a sampled homogeneous root is
obtained.  In each projective factor, choose a nonzero homogeneous coordinate,
set it equal to one, and let \(z=(z^1,\ldots,z^n)\) denote all remaining
ambient affine coordinates.
On a chart \(U\), we evaluate the regular local representatives of all
defining sections \(p_i\) and \(q_j\), using the explicit formulas in
Appendices~\ref{app:h11}--\ref{app:h22}, and collect them in the ordered
vector
\[
 F_U(z)=0
\]
of local defining equations.  We then form its \(c\times n\) Jacobian
\begin{equation}
 J=\frac{\partial F_U}{\partial z}.
 \label{eq:local-equation-jacobian}
\end{equation}
Here \(c=n-3\), since \(X\) is a threefold.

Let \(\mathsf D\) index \(c\) dependent coordinates \(v=z_{\mathsf D}\), and
let \(\mathsf I\) index the remaining three coordinates
\(u=z_{\mathsf I}\), which serve as intrinsic coordinates.  With
\[
 J_{\mathsf D}=\frac{\partial F_U}{\partial v},\qquad
 J_{\mathsf I}=\frac{\partial F_U}{\partial u},
\]
the implicit-function theorem gives
\begin{equation}
 \frac{\partial v}{\partial u}=-J_{\mathsf D}^{-1}J_{\mathsf I}.
 \label{eq:implicit-tangent-solve}
\end{equation}
Equivalently, the ambient tangent map \(T=\partial z/\partial u\) is
\begin{equation}
 T_{\mathsf I}=\mathbf 1_3,\qquad
 T_{\mathsf D}=-J_{\mathsf D}^{-1}J_{\mathsf I},
 \qquad J T=0.
 \label{eq:ambient-tangent-map}
\end{equation}
For each choice of dependent coordinates, let
\(\sigma_{\min}(J_{\mathsf D})\) and \(\sigma_{\max}(J_{\mathsf D})\) denote the smallest and
largest singular values of \(J_{\mathsf D}\), and define its spectral condition number
by
\[
 \kappa_2(J_{\mathsf D})=
 \frac{\sigma_{\max}(J_{\mathsf D})}{\sigma_{\min}(J_{\mathsf D})}.
\]
Among the nonsingular choices of dependent coordinates, we select the block
\(J_{\mathsf D}\) with the smallest spectral condition number.  The numerical
nonsingularity threshold and the chart-consistency tests are reported in
Appendix~\ref{app:globaldata}.

Every ambient section derivative is then pulled back with the same tangent
map:
\begin{equation}
 \frac{\partial s_U}{\partial u^a}
 =\sum_{\mu=1}^{n}
   T_{\mu a}\frac{\partial s_U}{\partial z^\mu}.
 \label{eq:intrinsic-section-jet}
\end{equation}
These values and first derivatives are the section jets supplied to the TN
contraction in Sec.~\ref{sec:method}.  Likewise, the Poincar\'e-residue
construction \cite{griffiths1969periods,braun2008} gives, in the same
dehomogenized frame,
\begin{equation}
 \Omega_U=h_u\,\dd u^1\wedge\dd u^2\wedge\dd u^3,
 \qquad
 \rho_{\Omega,U}=|h_u|^2=\frac{1}{|\det J_{\mathsf D}|^2}.
 \label{eq:computational-residue}
\end{equation}
The fixed projective and equation-frame factors are already included in
the chosen local representatives; the ordering-dependent sign of \(h_u\)
drops out of \(\rho_{\Omega,U}\).

At each sampled point the geometric construction provides the section values
and first derivatives, the Calabi--Yau and proposal densities, the importance
weight and the tangent map.  Together with the fixed or learned local
operators \(Q_\alpha\), these data determine the feature derivatives
\begin{equation}
 \partial_a(s_U^\dagger Q_\alpha s_U),\qquad
 \partial_a\bar\partial_{\bar b}
 (s_U^\dagger Q_\alpha s_U),
 \label{eq:feature-jets}
\end{equation}
and combines them into
\((F_\theta,\partial F_\theta,\partial\bar\partial F_\theta)\).  Thus the atlas, tangent solve,
source jets, the local coefficient \(h_u\) of \(\Omega\), and the integration
measure are determined entirely by the
geometric construction; the TN enters only through the algebraic
contraction and metric derivatives.

The TN contraction and the transition and residue formalism are common to
the three examples.  Applying them to a new gCICY still requires regular
local representatives, a source-section basis and a suitable sampling
construction for that geometry.

The covariance derived in
Eqs.~\eqref{eq:conormal-change}--\eqref{eq:residue-density-change} has a
direct numerical consequence: changing the projective chart or the
nonsingular block \(J_{\mathsf D}\) transforms \(s_U\), \(T\) and \(h_u\)
but leaves the projective source map, \(\det g/\rho_\Omega\) and the
normalized importance weight unchanged.  Atlas gluing is therefore part of
the geometric construction rather than the TN ansatz.  The explicit local
representatives, residue covariance, sampling measure, source-space ranks and
numerical chart comparisons are collected in
Appendices~\ref{app:members}--\ref{app:immersion}.

\section{Numerical results for three gCICY threefolds}
\label{sec:results}

In this section we put the construction in the last section to work.  Section~\ref{sec:training} introduces how
we measure the error of a computed metric, what each model minimizes, and
how two metrics are allowed
to be compared; the tests on three geometries then follow, one at a time.  On
\(X_{11}\) we train three refinements of one common degree-two
metric---the TN, a parameter-matched neural correction reading the same
section data, and an unrestricted \(H_6\) metric---and compare all
nine runs on a single final sample.  On \(X_{21}\) we ask how much
accuracy a degree-eight TN gains over the degree-four unrestricted Hermitian metric
while using fewer parameters, and whether raising the bond dimension or
the degree reproducibly helps.  On \(X_{22}\) we run the same construction
on a two-stage geometry, where an equal-update comparison identifies the
truncated local operator dictionary as the limiting component at fixed
bond dimension.  The reported errors are computed with common pointwise
weights within each comparison and are not compared across geometries.

\subsection{Error measures, loss functions, and comparison design}
\label{sec:training}

\paragraph{Error measures.}

All models are judged by how far their volume ratio departs from that of
the Ricci-flat metric; the loss minimized during training is built from
the same ratio and is specified below.  In local intrinsic coordinates,
define
\begin{equation}
 \eta_\theta(x)=\frac{\det g_\theta(x)}{\rho_\Omega(x)},\qquad
 r_\theta(x_i)=\frac{\eta_\theta(x_i)}
 {\sum_jw_j\eta_\theta(x_j)}.
 \label{eq:ratio}
\end{equation}
For brevity, write \(r_i=r_\theta(x_i)\).
The normalization enforces \(\sum_iw_ir_i=1\) on the evaluation sample.
It removes one overall scale but does not enforce \(r_i=1\) pointwise.
For the exact Ricci-flat metric \(r\equiv1\) identically, so the
deviation of \(r_i\) from one measures the pointwise mismatch between the
model volume form and the Calabi--Yau volume form.  All statistics below
are therefore evaluated with respect to the normalized Calabi--Yau measure
\(\dd\mu_\Omega/\int_X\dd\mu_\Omega\), rather than by assigning equal
weight to the sampled points.  We report
\begin{align}
 \sigma&=\sum_iw_i|1-r_i|,\label{eq:sigma}\\
 \chi&=\left[\sum_iw_i(1-r_i)^2\right]^{1/2},\label{eq:chi}\\
 a_i&=|\log r_i|,\qquad
 u_i=(\log r_i)_+,\qquad
 d_i=(-\log r_i)_+,\label{eq:bilateral-tail}\\
 Q^{\rm abs}_{0.999}&=
 \min\Bigl\{t:\textstyle\sum_{i:\,a_i\le t}w_i\ge0.999\Bigr\},
 \label{eq:q999}\\
 \CVaR^{\rm abs}_{1\%}&=
 \frac{1}{\alpha}\left(
 \sum_{j<J}w_{(j)}a_{(j)}+\delta a_{(J)}\right),
 \qquad \alpha=0.01.
 \label{eq:cvar}
\end{align}
Here \((x)_+=\max\{x,0\}\).
The first two are bulk averages of the volume-ratio error.  The
logarithmic variables treat over- and underestimation symmetrically and
separate the two failure directions: \(u\) records local volume excess
(\(r>1\)) and \(d\) local volume deficit (\(r<1\)), the direction in which
the metric approaches degeneracy; both tails are therefore reported.
Equation~\eqref{eq:q999} is the weighted \(0.999\)-quantile of \(a\):
outside a region carrying \(0.1\%\) of the manifold volume, the local
volume-form mismatch is at most a factor \(e^{Q^{\rm abs}_{0.999}}\).
For the last definition, the indices are ordered so that
\(a_{(1)}\geq a_{(2)}\geq\cdots\), \(J\) is the first index for which
\(\sum_{j\leq J}w_{(j)}\geq\alpha\), and
\(\delta=\alpha-\sum_{j<J}w_{(j)}\).  Thus the boundary observation is
included fractionally when its weight crosses the one-percent threshold.
The resulting \(\CVaR^{\rm abs}_{1\%}\) is exactly the weighted mean of
\(a\) over the worst one percent of the manifold volume; the
quantile locates the tail, the CVaR measures its average depth, and the
same summaries with \(a\) replaced by \(u\) and \(d\) resolve the two
sides.  All results in this paper use this two-sided convention.  We
additionally report the observed range of \(r\) and the minimum metric
eigenvalue; the latter is computed in the deterministic implicit chart of
each point; its positivity is chart-independent, while its magnitude
depends on that chart and is an implementation check rather than a
geometric invariant.

Finally, two effective sample sizes record how much statistical power the
weighted sample actually carries.  The importance effective sample size
\begin{equation}
 \ESS=\frac{1}{\sum_iw_i^2}
 \label{eq:ess}
\end{equation}
is the number of equally weighted draws to which the sample is
equivalent: it equals the number of points for uniform weights and
decreases as the weights become uneven.  When the sampled points are
grouped into complete fibres \(C\), with total cluster weight
\(W_C=\sum_{i\in C}w_i\), the corresponding count of independent fibres is
\begin{equation}
 \ESS_{\rm fibre}=\frac{1}{\sum_C W_C^2}.
 \label{eq:fibre-ess}
\end{equation}
Because the several points of one fibre share a base point and are
correlated, \(\ESS_{\rm fibre}\), rather than \(\ESS\), controls the
resolution of the tail statistics: a \(0.999\) quantile rests on roughly
\(10^{-3}\,\ESS_{\rm fibre}\) independent clusters.

\paragraph{Loss functions.}

All reported TN optimizations and the unrestricted \(H_6\) runs in the
\(X_{11}\) comparison minimize the native geometric loss
\begin{equation}
 \mathcal L_{\rm native}
 =\mathcal L_{\log\MA}+\mathcal L_{\MA}
 +0.1\,\mathcal L_{\rm tail}.
 \label{eq:native-loss}
\end{equation}
The \(\phi\)-model uses the same geometric terms together with two
architecture-specific stabilizers, specified below.

Explicitly,
\begin{align}
 \widetilde\ell_i&=\log\det g_\theta(x_i)-\log\rho_\Omega(x_i)-\log\kappa_0,
 &\widetilde r_i&=e^{\widetilde\ell_i},\nonumber\\
 \mathcal L_{\log\MA}&=\sum_iw_i\widetilde\ell_i^2,
 &\mathcal L_{\MA}&=\sum_iw_i(\widetilde r_i-1)^2.
 \label{eq:native-components}
\end{align}
Equations~\eqref{eq:native-components} define the objective on the complete
training sample.  During stochastic optimization, the scripts restrict the
weights to the current minibatch and renormalize them to sum to one; the tail
mean is likewise evaluated within that minibatch.  This is a self-normalized
minibatch approximation, not an exactly unbiased stochastic gradient of the
complete-sample weighted sum.  The same convention is used within each
matched comparison.  Model selection and all reported observables are instead
computed on complete held-out samples with their globally normalized
importance weights.
Here the volume normalization \(\kappa_0\) is fixed once, from the initial
model evaluated on the training sample, so the training target does not
move.  The logarithmic term penalizes volume excess and deficit
symmetrically; the ratio term is the training counterpart of the reported
\(\chi\).
The tilded quantities are the fixed-normalization training quantities; in
particular, \(\widetilde r_i\) need not average to one on a later sample.  By
contrast, every final metric reported in
Eqs.~\eqref{eq:ratio}--\eqref{eq:chi} uses the untilded ratio \(r_i\),
renormalized on that evaluation sample.
The tail term targets rare spikes that the averages can hide.  With
\(r_\star=1.5\), \(\tau=0.05\) and tail mass
\(\beta=0.02\), it is
\[
 \mathcal L_{\rm tail}
 =\CVaR_\beta\!\left[
 \left(
 \tau\log\!\left(1+\exp\frac{\widetilde\ell_i-\log r_\star}{\tau}\right)
 \right)^{\!2}\,\right].
\]
Here \(\CVaR_\beta\) is the weighted tail mean of
Eq.~\eqref{eq:cvar}, with tail mass \(\beta\) in place of \(1\%\) and
applied to the bracketed quantity, a smoothed hinge that activates only
where the ratio exceeds \(r_\star\).  This one-sided training penalty is
distinct from the two-sided evaluation in
Eqs.~\eqref{eq:bilateral-tail}--\eqref{eq:cvar}.
Only when \(e^{\widetilde\ell_i}\) is formed inside the training loss, we
replace \(\widetilde\ell_i\) by
\(\min\{b,\max\{a,\widetilde\ell_i\}\}\).  We use
\((a,b)=(-20,20)\) for the TN and unrestricted Hermitian models and
\((-15,10)\) for the \(\phi\)-model.  This prevents numerical overflow; it
does not truncate \(\mathcal L_{\log\MA}\), and all reported observables use
the unrestricted ratio.
The \(\phi\)-model's positivity is not structural---a scalar correction
can drive a metric eigenvalue through zero---so its loss adds the two
stabilizers:
\begin{equation}
 \mathcal L_{\phi}=\mathcal L_{\rm native}
 +20\,\mathcal L_{\rm pos}+10^{-5}\mathcal L_{\rm corr}.
 \label{eq:phi-stabilizers}
\end{equation}
For eigenvalues \(\lambda_{ia}\), \(a=1,2,3\), of the model metric at the \(n\)
training points, and with
\(\bar\lambda_{H_2,i}=\frac13\operatorname{tr}g_{H_2}(x_i)\), these are
\begin{align}
 \mathcal L_{\rm pos}
 &=\frac{1}{3n}\sum_{i,a}\frac1{40}
 \log\!\left[1+
 \exp\!\left(40\left(0.01-\frac{\lambda_{ia}}
 {\bar\lambda_{H_2,i}}\right)\right)\right],\\
 \mathcal L_{\rm corr}
 &=\frac{\sum_i\|\partial\bar\partial\phi(x_i)\|_F^2}
 {\sum_i\|g_{H_2}(x_i)\|_F^2}.
 \label{eq:phi-stabilizer-components}
\end{align}
Here \(\|\cdot\|_F\) denotes the Frobenius norm,
\(\mathcal L_{\rm pos}\) is a positivity barrier at one percent of
the local \(H_2\) scale and \(\mathcal L_{\rm corr}\) weakly regulates
the correction size; the TN loss needs neither
(Proposition~\ref{prop:positive}).  Both affect the \(\phi\)-model's
optimization but are excluded from the common geometric validation quantity
\begin{equation}
 S_{\rm val}=
 \mathcal L_{\rm native}\big|_{\text{validation sample}},
 \label{eq:x11-validation-score}
\end{equation}
which is Eq.~\eqref{eq:native-loss} evaluated with the validation points and
weights.

\paragraph{Comparison design.}
\label{sec:metric-evaluation}

The source spaces and degrees are geometry-specific and were fixed in
Sec.~\ref{sec:geometry}; what is shared across geometries is how models
are initialized, compared and recorded.  Degree continuations that start from the
bond-dimension-one lift of a base metric \(H_b\) use the exact function-preserving
identity~\eqref{eq:exact-degree-continuation}; this covers the \(X_{11}\)
warm start below.  Continuations that start from an already trained TN
(the \(X_{21}\) degree ladder and the \(X_{22}\) degree step) instead
re-embed the fitted coefficients and preserve the represented function
only up to the reference term and the bond-rank constraints quantified in
the corresponding subsections below; they are initialization schemes, not
exact identities.
The reference-term scale \(\varepsilon_{\rm ref}\) is a numerical
setting rather than a universal constant, and rescaling it is absorbed by
the overall scale of \(B_\theta\).  When a bond is enlarged, the new block
on one side is set to zero and the block on the other side is nonzero; this
preserves the represented function while leaving new trainable parameters.
Both constructions are specified in
Appendix~\ref{app:eps-and-activation}, and results obtained at different
\(\varepsilon_{\rm ref}\) are never combined into a causal comparison.

The same bookkeeping applies to every comparison reported below.  The
data used to optimize the parameters, select the final metric and
evaluate the reported errors are kept separate, and whenever two metrics
are compared directly they are evaluated on exactly the same points and
with the same integration weights.  Results from exploratory
calculations are identified as such and are not used as independent
evidence for model ordering.

\subsection{\texorpdfstring{\(X_{11}\)}{X11}: metric comparison from a common initial metric}

On \(X_{11}\), we compare three metric constructions initialized from the
same degree-two metric: a positive degree-six TN, a parameter-matched neural
K\"ahler-potential correction, and an unrestricted \(H_6\) model.  Each model
is optimized three times independently, and the resulting metrics are
evaluated on the same independent sample.  We first compute a
positive Hermitian form \(H_2\) by the degree-two \(T\)-iteration with
\(\dd\mu=\dd\mu_\Omega\).  If
\(s_2\) denotes the complete 45-dimensional basis of
\(H^0(X_{11},L^2)\), the common initial metric is defined by
\begin{equation}
 F_2=s_2^\dagger H_2s_2,
 \qquad K_{H_2}=\frac12\log F_2.
 \label{eq:x11-common-h2}
\end{equation}
The first refinement is a three-site TN built from three copies of the
degree-two section vector.  It has algebraic degree \(k=6\), bond dimension
\(D=5\), the complete local operator basis \(q=45^2=2025\), and
\(141{,}750\) real trainable parameters.  By comparison,
\(h^0(X_{11},L^6)=871\), so an unrestricted \(H_6\) metric has \(871^2=758{,}641\)
real matrix entries; the TN uses \(18.7\%\) of that raw coefficient count
(Table~\ref{tab:parameter-comparison}).
Its potential is
\begin{equation}
 K_{\rm TN}=\frac16\log\!\left[
 \left\|B_\theta s_2^{\otimes3}\right\|^2
 +\varepsilon_{\rm ref}F_2^3\right].
 \label{eq:x11-tn-potential}
\end{equation}
The initial TN represents the exact power \(F_2^3\), and therefore gives
the same metric as Eq.~\eqref{eq:x11-common-h2}.  It is then enlarged
from \(D=1\) to \(D=5\) without changing the represented function: in
Eq.~\eqref{eq:function-preserving-bond-expansion}, the newly added block on one
side is set to zero while the block on the other side is nonzero.  All three
coefficient tensors are varied jointly.  Positivity follows from the
factorized square and the fixed positive reference term.

The second refinement follows the scalar-potential strategy used in neural
Calabi--Yau metric calculations~\cite{larfors2021size}:
\begin{equation}
 K_\phi=K_{H_2}+\phi_\vartheta.
 \label{eq:cymetric-residual}
\end{equation}
To give this model the same degree-two section information used by the TN,
we form the projectively invariant normalized section outer product
\begin{equation}
 \rho^{(s_2)}_{ab}
 =\frac{(s_2)_a\overline{(s_2)_b}}{s_2^\dagger s_2},
 \qquad a,b=1,\ldots,45.
 \label{eq:x11-residual-features}
\end{equation}
The diagonal entries and the real and imaginary parts above the diagonal
provide \(45^2=2025\) real input features.  A fully connected network with
trainable parameters \(\vartheta\) and three width-66 GELU hidden layers
produces the real scalar
\(\phi_\vartheta\); its raw output is multiplied by \(0.1\) before it is
added to the K\"ahler potential.  The output layer has no bias and is initialized to
zero, so this model also begins with \(K_\phi=K_{H_2}\).  It has
\begin{equation}
 P_\phi=(2025w+w)+2(w^2+w)+w=142{,}626,
 \qquad w=66,
 \label{eq:x11-residual-count}
\end{equation}
real trainable parameters.  The projectively invariant input
makes \(\phi_\vartheta\) a global scalar correction and therefore preserves
the K\"ahler class, while positivity of
\(g_{H_2}+\partial\bar\partial\phi_\vartheta\) is encouraged by an
eigenvalue penalty during training rather than guaranteed by the
parameterization.  This is our \(\phi\)-model adapted to the
gCICY geometry, not a direct execution of the public
\texttt{cymetric} package~\cite{larfors2022}.

Equality of the initial metrics and source data does not make the two
parameter-space initializations identical.  The TN begins from a
structured factorization that exactly represents the degree-two metric,
whereas the hidden weights of the \(\phi\)-model are random and
only its output layer is zero.  The ability to embed a lower-degree metric
exactly into a higher-degree factorization is part of the algebraic
parameterization being tested; the comparison does not claim to isolate it
from all other optimization effects.
Both models read the same complete degree-two section data: the
TN contracts three copies of the section vector through a positive
algebraic factorization, while the \(\phi\)-model flattens the normalized
outer product in Eq.~\eqref{eq:x11-residual-features} into a scalar
network.  The comparison therefore
evaluates the two complete metric constructions at a comparable parameter
budget.

To compare with the unrestricted algebraic family at the same degree,
we also trained three full-\(H_6\) models.  The numerically reconstructed
multiplication map from cubic products of degree-two sections has rank
\(871=h^0(X,L^6)\) at all three tested tolerances, so within numerical
precision these products span the complete degree-six section space.  Let
\(\mathcal M\) be the matrix expressing the ordered cubic products in the
chosen degree-six basis, so that
\(s_2^{\otimes3}=\mathcal M s_6\).  We initialize
\[
 H_6^{(0)}=c\,\mathcal M^\dagger H_2^{\otimes3}\mathcal M,
 \qquad c>0.
\]
This matrix is positive definite because \(H_2\) is positive definite and
\(\mathcal M\) has full column rank.  Moreover,
\[
 s_6^\dagger H_6^{(0)}s_6
 =c\bigl(s_2^\dagger H_2s_2\bigr)^3=cF_2^3.
\]
Consequently, \(H_6^{(0)}\) gives exactly the same K\"ahler metric as the
common \(H_2\) starting point; the constant \(c\) only fixes the matrix
normalization.
The product-section relation and the equality of the lifted metric with the
degree-two source were also verified numerically, as reported in
Appendix~\ref{app:x11-settings}.  Starting from this lift, without using a
fitted TN, we optimize \(H_6=S_6^\dagger S_6\) to validation convergence on
the same training and validation samples.  The resulting metrics characterize
the solutions reached from the direct \(H_2\) lift; they do not estimate the
global minimum of the unrestricted family.

Each architecture is optimized three times from independent random
initializations and selected with the common geometric validation quantity
in Eq.~\eqref{eq:x11-validation-score}.  The nine resulting metrics are then
evaluated on one independent \(200{,}000\)-point sample with
\(\ESS=101{,}536\).  The comparison concerns validation-converged solutions
rather than equal computational cost; sample sizes and optimization settings
are given in Appendix~\ref{app:x11-settings}.
Table~\ref{tab:x11-comparison} reports the common-sample results.
\begin{table}[!htbp]
\centering
% Generated from frozen pointwise X11 arrays by
% scripts/generate_gcicy_tn_final_tables.py.
\resizebox{\textwidth}{!}{%
\begin{tabular}{@{}lrrrr@{}}
\toprule
model & trainable real parameters & $\sigma$ & $\chi$ & $\lambda_{\min}(g)$ \\
\midrule
common $H_2$ start & 2,025 & $0.25164$ & $0.30889$ & $0.01817$ \\
positive $k=6,D=5$ TN & 141,750 & $0.01480\pm0.00138$ & $0.02260\pm0.00161$ & $0.02502\pm0.00018$ \\
neural $\phi$-model & 142,626 & $0.02237\pm0.00025$ & $0.03307\pm0.00035$ & $0.02606\pm0.00030$ \\
direct $H_6$, extended to plateau & 758,641 & $0.02912\pm0.00054$ & $0.05583\pm0.00124$ & $0.02520\pm0.00005$ \\
\bottomrule
\end{tabular}
}

\begin{tabular}{@{}lrrr@{}}
\toprule
model & $Q_{0.999}(|\log r|)$ & $\operatorname{CVaR}_{1\%}(|\log r|)$ & $\max |\log r|$ \\
\midrule
positive $k=6,D=5$ TN & $0.20437\pm0.01650$ & $0.12663\pm0.00553$ & $0.65712\pm0.02200$ \\
neural $\phi$-model & $0.25888\pm0.01130$ & $0.17803\pm0.00441$ & $0.57768\pm0.05859$ \\
direct $H_6$, extended to plateau & $0.54807\pm0.01094$ & $0.37578\pm0.01508$ & $0.99497\pm0.03681$ \\
\bottomrule
\end{tabular}

\caption{\(X_{11}\) errors on a common independent \(200{,}000\)-point
sample.  Entries are means and sample standard deviations over three
optimizations.  The TN and \(\phi\)-model start from the same \(H_2\) metric;
the unrestricted \(H_6\) model starts from its exact cubic lift.  Tail
quantities use \(|\log r|\), and all sampled metrics are positive.}
\label{tab:x11-comparison}
\end{table}
Both models improve the common starting metric by more than an order of
magnitude: \(\sigma\) falls from \(0.25164\) to \(0.01480\pm0.00138\) for
the TN and to \(0.02237\pm0.00025\) for the \(\phi\)-model.  The two
sets of runs do not overlap---the worst TN run, \(\sigma=0.01585\), is
still better than the best \(\phi\)-model run, \(0.02209\)---so the ordering
does not depend on which pair of runs is compared.  On the run means the
TN lowers \(\sigma\) by \(33.8\%\) and \(\chi\) by \(31.7\%\).

Resampling complete fibres within each run confirms these gains: the
paired intervals stay positive in all three runs for \(\sigma\), for
\(\chi\), and for the tail conditional mean \(\CVaR_{1\%}(|\log r|)\),
which the TN lowers by \(28.9\%\).  The extreme quantile
\(Q_{0.999}(|\log r|)\) is \(21.1\%\) lower on average and lower in each
of the three runs, but one paired interval reaches zero; we therefore
report it as an aggregate improvement rather than a run-by-run one.
Figure~\ref{fig:x11-tail-survival} shows the corresponding tail survival
functions: the TN curve lies below the \(\phi\)-model curve over the resolved
part of both tails, and the two converge only in the sparsely populated
extreme where the quantile estimate loses resolution.
\begin{figure}[!htbp]
\centering
\includegraphics[width=\textwidth]{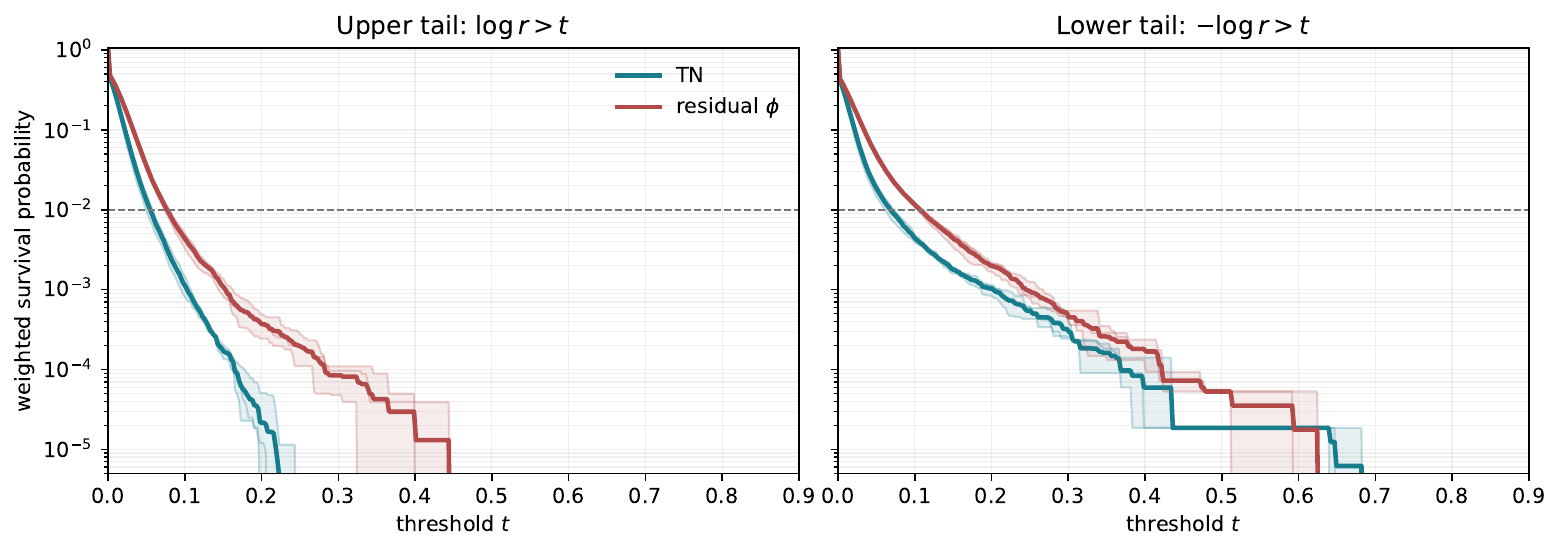}
\caption{Weighted upper- and lower-tail survival functions for the TN and
the \(\phi\)-model on the common \(X_{11}\) final sample.
Thin curves show the three independent runs; thick curves show their
pointwise mean, and shaded regions their range.  The dashed line marks
one-percent tail mass.  The TN has lower \(\CVaR_{1\%}(|\log r|)\) in all
three paired comparisons.  The mean \(Q_{0.999}(|\log r|)\) is also lower,
although that quantile is unresolved in the third pair.  The plot is an
empirical statement on the sampled measure, not a pointwise dominance
claim or a uniform bound.}
\label{fig:x11-tail-survival}
\end{figure}

Because the initial metrics coincide exactly, these differences cannot be
attributed to a better starting geometry.  They show that, under the
common geometric validation criterion, the positive algebraic TN uses this
parameter budget more effectively than the tested
\(\phi\)-model.  This is a comparison of the computed metrics under the stated
optimization procedures, not a claim about the global minimum of either
model family.

Direct optimization of the unrestricted \(H_6\) family from the exact
\(H_2\) lift reaches
\((\sigma,\chi)=(0.02912\pm0.00054,\,0.05583\pm0.00124)\).  The TN means
are lower by \(49.2\%\) and \(59.5\%\), respectively, while its mean
\(Q_{0.999}(|\log r|)\) and \(\CVaR_{1\%}(|\log r|)\) are lower by
\(62.7\%\) and \(66.3\%\).  Using the Hermitian forms represented by the
optimized TNs as initial conditions for unrestricted \(H_6\) optimization
lowers the validation errors further to
\(\sigma=(0.00797,0.00779,0.00801)\).  The two representations agree at the
metric level to approximately \(10^{-9}\).  Since the unrestricted family
contains the TN family, this contrast demonstrates sensitivity to
initialization and conditioning rather than greater expressivity of the
factorized ansatz.

We also asked whether a TN with a raw parameter count close to that of
unrestricted \(H_6\) could close the remaining gap.  This was a development study,
performed after the principal \(X_{11}\) comparison had been inspected,
and it does not enter the final-sample ranking.  Five allocations with
\((k,D)=(6,12),(8,9),(14,6),(18,5),(26,4)\) contain between \(680{,}400\)
and \(777{,}600\) parameters, against \(758{,}641\) for unrestricted \(H_6\); their
common development-evaluation errors are
\(\sigma_{\rm dev}=0.014112\), \(0.010475\), \(0.010734\), \(0.010706\)
and \(0.010854\), respectively.  Continuing the fitted \(k=6,D=5\) model
to \(k=6,D=13\), with \(789{,}750\) parameters, reaches
\(\sigma_{\rm dev}=0.009236\) under a gentler learning-rate schedule,
whereas the same final architecture reaches \(0.013644\) under the
original schedule.  On the common validation sample the gentler TN path
has \((\sigma,\chi)=(0.008250,0.013240)\), while the corresponding
TN-initialized unrestricted \(H_6\) continuation has \((0.007970,0.013567)\).  The
near-budget TN therefore approaches the parent-family result without
dominating it: its \(\sigma\) is higher and its \(\chi\) slightly lower.
The schedule comparison also shows that the learning rate alone can move
\(\sigma_{\rm dev}\) by roughly a third at fixed architecture, which we
return to in Sec.~\ref{sec:conclusion}.  The full architecture and path
tables, together with the fitted checkpoints and pointwise evaluation arrays,
are included in the reproduction archive~\cite{reproarchive}.

\subsection{\texorpdfstring{\(X_{21}\)}{X21}: compression and the limits of bond growth}

For \(X_{21}\), the complete degree-one source has dimension \(d=11\),
and hence the local operator basis contains \(q=d^2=121\) elements.  We use this
geometry to answer two separate questions: can the factorized form reach a
more accurate metric than a full Hermitian form at a comparable parameter
budget in some situations, and does enlarging the bond dimension of an already trained
network actually help?

\paragraph{Structured compression at degree eight.}
We compare a degree-four metric with \(104{,}976\) real parameters and a
degree-eight, \(D=8\) tensor network with \(96{,}800\) parameters.  The TN
entry is the median-validation member of three optimizations.  Both metrics
are evaluated on the same independent \(196{,}608\)-point sample, for which
\(\ESS=137{,}142\) and the \(32{,}768\) complete fibres give
\(\ESS_{\rm fibre}=29{,}311\).

\begin{table}[t]
\centering
\small
% Generated from frozen pointwise X21 arrays by
% scripts/generate_gcicy_tn_final_tables.py.
\begin{tabular}{lrcccc}
\toprule
model & \(P\) & \(\sigma\) & \(\chi\) & \(Q^{\rm abs}_{0.999}\) & \(\CVaR^{\rm abs}_{1\%}\) \\
\midrule
full \(H_4\) & 104,976 & 0.016273 & 0.023968 & 0.206 & 0.127 \\
positive \(k=8,D=8\) TN & 96,800 & 0.010503 & 0.015236 & 0.129 & 0.076 \\
\bottomrule
\end{tabular}

 \caption{The unrestricted degree-four Hermitian metric and the positive
 degree-eight tensor network on a common independent \(X_{21}\) sample.  The
 column \(P\) gives the number of real trainable parameters.  Tail quantities
 use \(|\log r|\); one optimized metric per representation enters this
 comparison.}
\label{tab:x21-final-blind-compression}
\end{table}

The tensor network gives the more accurate metric across the board:
\(\sigma\) falls from \(0.016273\) to \(0.010503\) and \(\chi\) from
\(0.023968\) to \(0.015236\)---reductions of \(35.46\%\) and
\(36.43\%\) with \(7.79\%\) fewer parameters---and the two tail
measures drop by roughly \(37\%\) and \(40\%\)
(Table~\ref{tab:x21-final-blind-compression}).  Resampling complete fibres
keeps all four improvements strictly positive for this pair of metrics.
Because only one optimized metric per representation enters, the comparison
does not measure between-fit variation; the saving concerns parameter count,
not evaluation or training time.

This ordering may appear to conflict with the \(X_{11}\) result, where
the best observed metric was ultimately an unrestricted \(H_6\) form warm-started
from the TN.  The two results instead illustrate complementary forms of
practical access supplied by the factorization, rather than greater
same-degree expressiveness.  On \(X_{11}\), the unrestricted \(H_6\) family is
computationally affordable, but neither direct route reaches the best
observed endpoint: direct unrestricted \(H_6\) optimization stalls, the TN
identifies a useful basin, and continuation in the parent family lowers
the error further.  Among the optimization paths tested here, that
endpoint is reached only through the TN initialization.  On \(X_{21}\),
the same-degree parent unrestricted \(H_8\) family contains \(5{,}953{,}600\)
real parameters and lies far outside the matched budget.  The TN instead
provides a trainable structured subfamily at degree eight, and its
comparison with the affordable unrestricted \(H_4\) model shows that, under the
tested finite budget, allocating parameters to degree through
compression is more effective than using the same budget for unrestricted
degree-four coefficients.  At equal degree the full family can refine a TN start; at
equal budget compression buys degree.  In both cases, the factorization
makes selected high-degree algebraic metrics practically accessible.

\paragraph{Bond growth at equal updates.}
In the early \(X_{21}\) runs, enlarging a trained \(D=8\) network to
\(D=12\) and training it further lowered the error, which suggests that
the larger bond helps.  But that workflow changes two things at once:
the model gets bigger, and it also trains longer.  To find out which one
lowers the error, we took the three fitted \(D=8\) models and continued
each along two branches: one keeps \(D=8\); the other first expands to
\(D=12\) without changing the represented function and then trains on.
Both branches use the same data, minibatch order, optimizer settings and
update count (Appendices~\ref{app:eps-and-activation}
and~\ref{app:x21-x22-settings}), so within each pair the only difference
is whether the extra bond channels exist.

Three numbers tell the story.  The sources start at a mean
\(\sigma=0.010505\); training alone brings this to \(0.008394\);
training with the enlarged bond brings it to \(0.008228\).  Nearly all
of the improvement therefore comes from the additional training.  The
extra channels move \(\sigma\) by a further \(1.92\%\pm2.86\%\) and
\(\chi\) by \(1.10\%\pm2.66\%\) across the three pairs---within the
run-to-run scatter, with one pair going the other way and no improvement
in the tails.  At equal update counts this experiment therefore cannot
resolve any benefit from the larger bond, even though the \(D=12\) model
carries \(2.22\) times as many parameters and a costlier update.  The
capacity is real; the optimizer did not turn it into accuracy here, and
the early \(D=8\to12\) gain was mostly continuation.  (The starting
value is the mean of the three source models on a fresh sample, close to
but distinct from the single checkpoint of
Table~\ref{tab:x21-final-blind-compression}; the reproduction
archive~\cite{reproarchive} contains the run-by-run values and the
fibre-cluster bootstrap records.)

\paragraph{Fixed-bond degree refinement.}
The complementary question is whether raising the algebraic degree
keeps paying off.  It matters for the ansatz itself: at fixed \(q=121\)
and \(D=8\) the parameter count grows only linearly in \(k\), so if
accuracy kept improving with degree, one could simply keep raising it.
For each \(k=8,12,16,20\) we trained three models, all initialized by the
registered repeat re-embedding of the same fitted \((4,5)\) checkpoint,
followed by function-preserving bond enlargement.  As stated in
Sec.~\ref{sec:metric-evaluation}, this cross-degree re-embedding is an
initialization scheme rather than the exact identity
\eqref{eq:exact-degree-continuation}.  Each model was stopped at the registered
validation plateau, frozen, and then evaluated blind on one new common
\(196{,}608\)-point sample.  The trainable count rises linearly along
the ladder, from \(96{,}800\) at \(k=8\) to \(282{,}656\) at \(k=20\).
(The \(k=8\) entry is a fresh continuation under this protocol, not the
Table~\ref{tab:x21-final-blind-compression} checkpoint retrained.)
These models are defined on the product-generated spaces
\(W_k=\operatorname{im}\mu_m\).  Since surjectivity is not established here
at these higher degrees, we do not identify \(W_k\) with the complete
space \(H^0(X,L^k)\).
\begin{table}[t]
\centering
% Generated from 12 frozen pointwise X21 arrays by
% scripts/generate_gcicy_tn_final_tables.py.
\small
\begin{tabularx}{\textwidth}{@{}rrXXXX@{}}
\toprule
\(k\) & trainable \(P\) & \(\sigma\) & \(\chi\) & \(Q^{\rm abs}_{0.999}\) & \(\CVaR^{\rm abs}_{1\%}\) \\
\midrule
8 & 96,800 & \(0.008767\pm0.001181\) & \(0.012791\pm0.001516\) & \(0.09823\pm0.00976\) & \(0.06597\pm0.00686\) \\
12 & 158,752 & \(0.005696\pm0.000232\) & \(0.008715\pm0.000264\) & \(0.07022\pm0.00061\) & \(0.04715\pm0.00066\) \\
16 & 220,704 & \(0.006357\pm0.000746\) & \(0.009318\pm0.000997\) & \(0.07058\pm0.00478\) & \(0.04756\pm0.00405\) \\
20 & 282,656 & \(0.006073\pm0.000158\) & \(0.009268\pm0.000199\) & \(0.07273\pm0.00358\) & \(0.04978\pm0.00142\) \\
\bottomrule
\end{tabularx}

\caption{Fixed-\(D=8\) degree comparison on the common independent
\(X_{21}\) final sample.  Entries are the mean and sample standard
deviation over three independently optimized models.  All use the complete
\(q=121\) local operator dictionary, and all sampled metrics are positive.
The tail quantities use \(|\log r|\).}
\label{tab:x21-degree-ladder}
\end{table}
Table~\ref{tab:x21-degree-ladder} shows a large initial gain that then
stops.  Comparing the three-run means, the step from \(k=8\) to \(12\)
lowers \(\sigma\) by \(35.0\%\), \(\chi\) by \(31.9\%\), and both tail
measures by \(28.5\%\).  An unpaired run-level comparison supports this
initial step.  The models at different degrees were trained independently;
their common final sample pairs only the Monte Carlo evaluation, not the
optimization runs.  Later mean changes are smaller than the observed
between-run variation: \(k=16\) is \(11.6\%\) worse in \(\sigma\) than
\(k=12\), while the \(k=16\to20\) step lowers the mean \(\sigma\) by
\(4.5\%\) and \(\chi\) by \(0.5\%\) but worsens the two tail means by
\(3.0\%\) and \(4.7\%\).  The data therefore establish the initial
degree gain but no reproducible improvement beyond \(k=12\); the largest
model is not the best observed endpoint.

This is a finite-range comparison rather than an asymptotic degree law.
Within the tested range, increasing the degree is useful from \(k=8\) to
\(k=12\), but linear parameter growth at fixed \(D\) does not by itself
produce monotone improvement.

\paragraph{Path dependence in the degree--bond plane.}
The preceding comparisons vary degree and bond dimension separately.
When both are allowed to change, the optimized metric also depends on the
initialization path through the \((k,D)\) plane.
Figure~\ref{fig:x21-path-dependence} compares models descended from the
same fitted \((k,D)=(4,5)\) metric, either by direct cross-degree
re-embedding followed by bond enlargement or by enlarging an intermediate
optimized model.

\begin{figure}[t]
\centering
\includegraphics[width=0.92\textwidth]{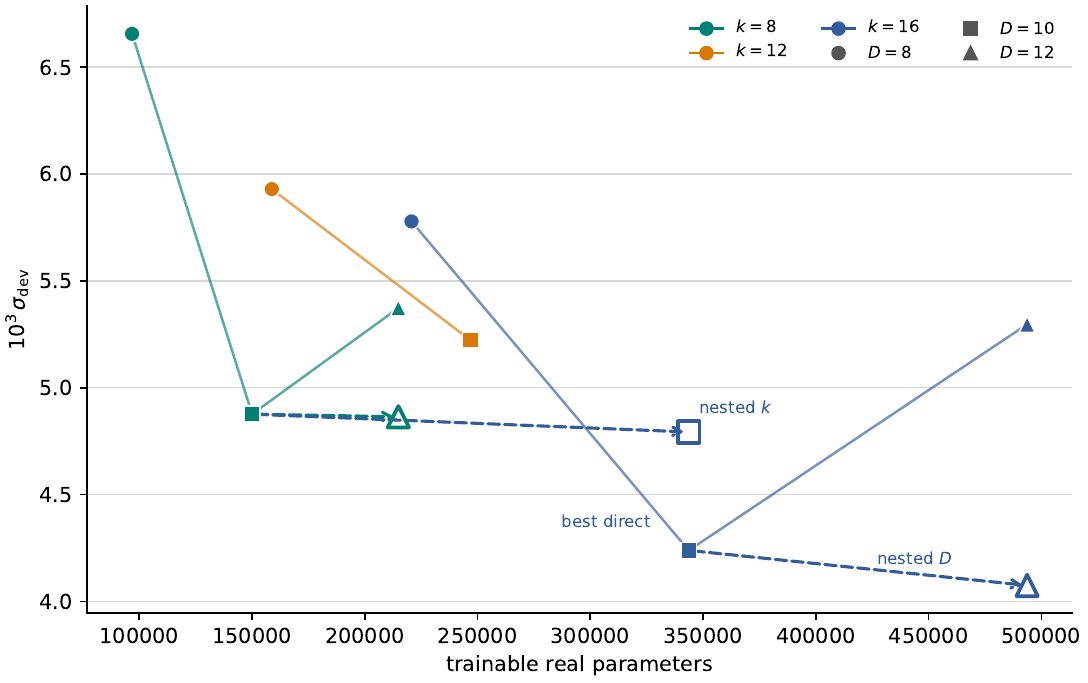}
\caption{Dependence of the \(X_{21}\) validation error on algebraic degree,
bond dimension and initialization path.  The vertical axis is
\(10^3\sigma\) on a common validation sample.  Filled points are optimized
directly from the \((4,5)\) metric; open points and dashed arrows follow
continuations from an already optimized model.  This exploratory comparison
maps the optimization landscape and is not a replicated architecture
comparison.}
\label{fig:x21-path-dependence}
\end{figure}

Among the direct optimizations, the lowest validation error occurs at
\((k,D)=(16,10)\), rather than for the largest model:
\begin{equation}
 (\sigma,\chi)=(0.004239,0.006692).
 \label{eq:x21-best-direct}
\end{equation}
Increasing its bond dimension to \(D=12\) and continuing the optimization
gives the lowest value observed in this study,
\begin{equation}
 (\sigma,\chi)=(0.004075,0.006446),
 \label{eq:x21-best-nested}
\end{equation}
whereas direct optimization of the same \((16,12)\) architecture from the
common origin performs substantially worse.  Continuation is not uniformly
advantageous: raising the degree of the fitted \((8,10)\) model improves it
only slightly and does not reach the directly optimized \((16,10)\) value.
At fixed degree the trainable families are exactly nested in the bond
dimension, so their exact optimal errors cannot increase with \(D\); the
different trained endpoints at the same \(k\) therefore expose path and
optimization dependence.  Across degrees, however, the re-embeddings of
Sec.~\ref{sec:metric-evaluation} preserve the function only approximately
and change the representation and optimization path together.  The
cross-degree nonmonotonicity does not by itself separate those effects, and
we make no capacity claim from it.
The near-budget \(X_{11}\) study gives the complementary outcome: there a
direct \(D=13\) continuation outperforms the tested nested schedules,
whereas here the \(D=10\to12\) continuation improves on the best direct
optimization.  Exact nesting at fixed degree guarantees inclusion of the
smaller variational family, but not that a particular optimization path
will find its best metric.

The degree-four baseline uses the complete section space: the multiplication
map has numerical rank \(324=h^0(X,L^4)\).  Its section basis therefore spans
all of \(H^0(X,L^4)\), rather than a proper product-generated subspace.
Surjectivity is not assumed for the higher-degree TN models.

\subsection{\texorpdfstring{\(X_{22}\)}{X22}: transfer to a two-stage geometry}

The final example tests the same positive tensor-network construction on the
type-\((2,2)\) geometry, whose second stage carries two generalized
defining sections.  The source space has dimension \(d=17\), the
complete local operator basis contains \(d^2=289\) elements, and all
models use bond dimension \(D=5\), with numerical settings in
Appendix~\ref{app:x21-x22-settings}.  We do two things here: compare
truncated learned dictionaries with the complete fixed one at degree
four, and then continue the fitted model to degree six.

\paragraph{Local-dictionary capacity at degree four.}

At degree four, we compare the fitted \(q=22\) dictionary, its
function-preserving enlargement to \(q=60\), and the complete fixed
dictionary under the matched optimization conditions of
Appendix~\ref{app:x21-x22-settings}.  For the complete model we use
all \(d^2=17^2=289\) matrix units as a fixed orthonormal operator basis
and train only their coefficient tensors.  The fitted \(q=22\) model
embeds into this basis with maximum potential and metric-entry
discrepancies \(1.6\times10^{-15}\) and \(1.1\times10^{-13}\),
respectively, before training.
Table~\ref{tab:x22-dictionary} reports the three-way comparison.
\begin{table}[t]
\centering
% Generated from frozen pointwise X22 arrays by
% scripts/generate_gcicy_tn_final_tables.py.
\begin{tabular}{lrcccc}
\toprule
local dictionary & \(P\) & \(\sigma\) & \(\chi\) & \(Q_{0.999}(|\log r|)\) & \(\CVaR_{1\%}(|\log r|)\) \\
\midrule
trainable \(q=22\) & 15,356 & 0.10040 & 0.13002 & 0.63609 & 0.50334 \\
trainable \(q=60\) & 41,880 & 0.07392 & 0.09740 & 0.54185 & 0.40997 \\
fixed complete \(q=289\) & 34,680 & 0.03853 & 0.05399 & 0.42533 & 0.26654 \\
\bottomrule
\end{tabular}

 \caption{Local-dictionary comparison for \(X_{22}\), using one optimized
 metric per dictionary on the same \(65{,}532\) points.  The \(q=22\) and \(q=60\) dictionaries
 are trainable, whereas the complete \(q=289\) matrix-unit dictionary is
 fixed; its parameter count therefore contains only coefficient tensors.
 The column \(P\) gives the number of real trainable parameters.  Tail
 statistics use \(|\log r|\).}
\label{tab:x22-dictionary}
\end{table}
Enlarging the learned dictionary from \(q=22\) to \(q=60\) already lowers
\(\sigma\) by \(26.4\%\).  Relative to that \(q=60\) dictionary, the fixed
complete basis lowers \(\sigma\), \(\chi\), \(Q^{\rm abs}_{0.999}\) and
\(\CVaR^{\rm abs}_{1\%}\) by a further \(47.9\%\), \(44.6\%\), \(21.5\%\)
and \(35.0\%\), respectively.  Using 2,000 paired resamples of the
21,844 fibre clusters, the central 95\% interval for each of the four
improvements lies strictly above zero.  The complete model also uses fewer
trainable parameters, \(34{,}680\) rather than \(41{,}880\): the learned
model spends \(82.8\%\) of its budget on the dictionary itself, and
since the complete basis spans \(\operatorname{End}(\CC^{17})\), its
family contains the learned one.  At this local dimension, truncating and
learning the operator dictionary therefore costs more parameters and gives
a less accurate metric than using the complete fixed basis.  This comparison
concerns local operator capacity, not bond dimension.

The learned-dictionary gain in Table~\ref{tab:x22-dictionary} relies on
the function-preserving expansion of
Appendix~\ref{app:eps-and-activation} applied to the dictionary: new
operators enter as a nonzero orthonormal complement with zero
coefficient slices, so the metric is unchanged while the new directions
remain reachable.  Useful capacity thus depends not only on the coefficient
count but on whether the added directions can be used by the optimization.
Each entry is one optimization evaluated on a common \(65{,}532\)-point
sample, with \(\ESS=23{,}399\) and fibre-cluster
\(\ESS_{\rm fibre}=9{,}547\).

\paragraph{Degree continuation.}

The same fitted \(q=60,k=4\) metric and its \(k=6\) continuation are
evaluated together on a second, independently generated
\(65{,}532\)-point sample with importance \(\ESS=22{,}579\) and
fibre-cluster effective sample size \(9{,}254\); the degree-four entries
here therefore need not coincide with those in
Table~\ref{tab:x22-dictionary}.  Continuing
the metric from degree four to degree six increases the trainable
count only from \(41{,}880\) to \(47{,}880\), or by \(14.3\%\), and gives
\begin{align}
 \sigma&:0.07447\to0.04738,&
 \chi&:0.10029\to0.06653,\nonumber\\
 Q^{\rm abs}_{0.999}&:0.91085\to0.60636,&
 \CVaR^{\rm abs}_{1\%}&:0.53377\to0.38537,
 \label{eq:x22-k}
\end{align}
with both entries evaluated on exactly the same points and with the same
importance weights.  The continuation reduces \(\sigma\) by
\(36.4\%\) and \(\chi\) by \(33.7\%\).  The upper-log quantile and
conditional tail mean improve from \(0.312\) and \(0.253\) to
\(0.213\) and \(0.172\), and their lower-log counterparts from
\(0.911\) and \(0.532\) to \(0.606\) and \(0.380\).  The observed
ratio range contracts from \([0.181,1.507]\) to \([0.237,1.412]\), and all
sampled metrics remain positive.  A paired bootstrap over complete fibre
clusters gives strictly positive improvement intervals for \(\sigma\),
\(\chi\), the absolute-log conditional tail mean and both one-sided
conditional tail means.  The interval for the \(0.999\) absolute-log
quantile crosses zero, so improvement of that particular quantile is not
treated as statistically resolved despite its lower point estimate.

\paragraph{Tail-statistic resolution.}

The two samples also quantify the resolution of the reported tail
statistics.  For the same fitted \(q=60,k=4\) metric, the degree-four and
degree-six samples give, respectively,
\[
 \begin{array}{c|cccc}
  &\sigma&\chi&Q^{\rm abs}_{0.999}&\CVaR^{\rm abs}_{1\%}\\
 \hline
 \text{degree-four sample}&0.07392&0.09740&0.54185&0.40997\\
 \text{degree-six sample}&0.07447&0.10029&0.91085&0.53377
 \end{array}
\]
Thus the bulk statistics drift by only \(0.7\%\) and \(3.0\%\),
whereas the two extreme-tail summaries drift by about \(68\%\) and
\(30\%\).  In particular, comparing the \(k=6\) value \(0.60636\)
with the degree-four-sample baseline \(0.54185\) would suggest that the
absolute-log quantile becomes worse, while the common-sample comparison
with \(0.91085\) suggests improvement.  We therefore use only the
same-sample paired comparison for the continuation and do not regard the
change in this quantile as resolved.

The matched continuation at fixed \(q=22\) shows that the improvements in
Table~\ref{tab:x22-dictionary} are not explained by additional optimization
alone.  By contrast, the degree-six calculation adds two sites and continues
the optimization, so it measures their combined effect rather than an
isolated degree contribution.
The training loss used a one-sided upper-tail term, whereas the
reported diagnostics are two-sided; this is why the lower-log tail is
displayed explicitly.

\section{Conclusion and outlook}
\label{sec:conclusion}

We have developed a positive tensor-network parameterization of algebraic
K\"ahler metrics and implemented it on three gCICY threefolds.  The
construction combines the tensor-network factorization with the geometric
data required on these spaces, including generalized sections,
Poincar\'e residues, intrinsic coordinate charts and fibrewise sampling.
For the source spaces used here, the associated projective maps are closed
immersions, so the resulting metrics are globally positive and remain in a
fixed K\"ahler class for every value of the trainable parameters.  At fixed
local and bond dimensions, both the parameter count and the measured
pointwise evaluation cost grow linearly with the number of multiplication
sites.

The \(X_{11}\) comparison shows that the factorization changes the
optimization landscape as well as the size of the representation.  Starting
from a common degree-two metric and using comparable parameter counts, the
TN consistently reaches lower bulk and tail errors than the neural
K\"ahler-potential correction.  Direct optimization of the unrestricted
degree-six Hermitian form from the same algebraic lift reaches less accurate
solutions, whereas unrestricted refinement initialized by a trained TN
produces the best metrics in this comparison.  Because the unrestricted
degree-six family contains the TN family, this ordering demonstrates
optimization-path dependence rather than an expressiveness advantage of the
factorized family at equal degree.

On \(X_{21}\), a degree-eight TN attains lower same-sample bulk errors than
the unrestricted degree-four metric while using slightly fewer trainable
parameters.  This result shows that, at a fixed parameter scale, compression
can make a higher algebraic degree more useful than an unrestricted
lower-degree form.  Increasing the bond dimension alone does not yield a
reproducible improvement, and the degree and bond studies show that the
observed gains eventually plateau and depend on initialization and
continuation.  The \(X_{22}\) calculation extends the same construction to a
two-stage gCICY and confirms that the required generalized-section and
residue data can be incorporated without changing the positive
tensor-network ansatz.

Taken together, these results identify the role of the TN as a structured
high-degree algebraic metric family rather than a universal replacement for
unrestricted Hermitian optimization.  The unrestricted family remains
powerful whenever it is affordable and can further refine a TN solution.
The benefit of factorization appears at degrees where the unrestricted
matrix becomes prohibitively large: for \(X_{21}\) at degree sixteen, the
unrestricted form has \(369{,}254{,}656\) real parameters, compared with
roughly \(4\times10^5\) in the networks studied here.  The calculations also
show that greater formal capacity is not sufficient by itself; the added
variational directions must also be explored effectively by the optimization.

Future work should test these conclusions on a broader collection of CICY
and gCICY geometries and develop continuation and adaptive-rank procedures
for which increasing the degree or bond dimension gives a reproducible gain
in accuracy.  A moduli-dependent extension, trained over complex-structure
and K\"ahler moduli rather than at a single fixed geometry, would turn the
present construction into a reusable family of metrics across moduli space.
The resulting metrics can then be assessed through physical observables,
including scalar and bundle Laplacian spectra, harmonic representatives,
matter-field normalizations and Yukawa couplings.

\section*{Acknowledgments and data availability}

The author used OpenAI Codex powered by ChatGPT models for software development
and assistance with model-training calculations, and Anthropic Claude and
DeepSeek models for manuscript editing and code review.  The author reviewed
all generated material and takes responsibility for the final scientific
content.  The source code is available on
\href{https://github.com/jtwangbimsa-dragon/gcicy-positive-tensor-networks}
{GitHub}.  Trained model parameters, evaluation samples, pointwise data and
scripts reproducing the tables and figures are archived on
Zenodo~\cite{reproarchive}.

\appendix

\section{Explicit generalized-section representatives}
\label{app:members}

\subsection{Type-\texorpdfstring{\((1,1)\)}{(1,1)} representatives}
\label{app:h11}

Let $m_{\rm H}$ denote the Hirzebruch-family index.  Then
\begin{equation}
 A=\PP^4_x\times\PP^1_y,\qquad
 p=\sum_{a=0}^{m_{\rm H}}g_a(x)y_0^{m_{\rm H}-a}y_1^a,\qquad
 \mathcal L=\cO_A(4,2-m_{\rm H}).
\end{equation}
The intermediate hypersurface is $M=\{p=0\}$.  Tensoring its restriction
sequence by $\mathcal L$ gives
\begin{equation}
 0\longrightarrow\cO_A(3,2-2m_{\rm H})
 \xrightarrow{\,\cdot p\,}\cO_A(4,2-m_{\rm H})
 \longrightarrow\cO_M(4,2-m_{\rm H})\longrightarrow0.
 \label{eq:hirzebruch-restriction-sequence}
\end{equation}
Since $H^0(A,\mathcal L)=0$, a generalized section is specified by a kernel class
\begin{equation}
 [c]\in\ker\!\left[
 H^1(A,\cO_A(3,2-2m_{\rm H}))\xrightarrow{\,\cdot p\,}H^1(A,\mathcal L)
 \right].
 \label{eq:hirzebruch-kernel-class}
\end{equation}
Cover $\PP^1_y$ by $U_0=\{y_0\neq0\}$ and $U_1=\{y_1\neq0\}$, and put
$t=y_1/y_0$ on the overlap.  There choose the C\v{e}ch
representative
\begin{equation}
 c=\sum_{j=1}^{2m_{\rm H}-3}C_j(x)t^{-j},
 \qquad \deg_x C_j=3.
 \label{eq:hirzebruch-cech-cochain}
\end{equation}
The class $[pc]$ vanishes precisely when $pc$ is a C\v{e}ch coboundary.
Splitting its Laurent series in $t$ gives the two chart representatives
explicitly: on $U_0$ the generalized equation is displayed by
$q_0=-(pc)_{\ge0}$, the part regular at $t=0$, and on $U_1$ by
$q_1=(pc)_{\le 2-m_{\rm H}}$, the part regular at $t=\infty$, each
converted to the corresponding homogeneous frame.  They satisfy
\begin{equation}
 q_1-q_0=pc\quad\text{on }U_0\cap U_1,
 \qquad q_0|_M=q_1|_M,
 \label{eq:hirzebruch-local-pieces}
\end{equation}
so the two displays agree on $M$ and differ off $M$ by a multiple of the
earlier equation $p$.

For the member used in this paper, $m_{\rm H}=3$, the target base degree
is $-1$ and $H^1(\PP^1,\cO(-1))=0$, so every class in
\eqref{eq:hirzebruch-kernel-class} maps to zero and the construction is
unobstructed.

\subsection{Type-\texorpdfstring{\((2,1)\)}{(2,1)} representatives}
\label{app:h21}

Write the two positive equations on $\PP^5_x\times\PP^1_y$ as
\begin{equation}
 p_1=y_0A_0(x)+y_1A_1(x),\qquad
 p_2=y_0^2B_0(x)+y_0y_1B_1(x)+y_1^2B_2(x),
\end{equation}
with $\deg A_i=1$ and $\deg B_i=2$.  Let $R_a$ be quadratic and $R_0,R_1$
linear in $x$.  Set $t=y_1/y_0$ on $U_0=\{y_0\ne0\}$ and
$u=y_0/y_1$ on $U_1=\{y_1\ne0\}$.  The regular frame representatives of
the $\cO(3,-1)$ section are
\begin{align}
 q_0&=A_1R_a+B_1R_1+B_2(R_0+tR_1)&&\text{on }U_0,\label{eq:x21-q0}\\
 q_1&=-A_0R_a-B_0(uR_0+R_1)-B_1R_0&&\text{on }U_1.\label{eq:x21-q1}
\end{align}
On the overlap, substituting the two affine ratios and clearing denominators
gives the exact polynomial identity
\begin{equation}
 y_0y_1(y_1q_0-y_0q_1)
 =y_0y_1p_1R_a+y_0p_2R_0+y_1p_2R_1.
 \label{eq:x21-overlap-identity}
\end{equation}
Hence $q_0=(y_0/y_1)q_1$ on $p_1=p_2=0$, which is exactly the
$\cO(3,-1)$ frame transition.  Differentiating
\eqref{eq:x21-overlap-identity} on $X$ gives the block-triangular equation
change required in \ref{lem:residue-covariance}.

\subsection{Type-\texorpdfstring{\((2,2)\)}{(2,2)} representatives}
\label{app:h22}

For $A=\PP^1_x\times\PP^1_y\times\PP^5_z$, write the positive equation
$p_2$ in both ways
\begin{equation}
 p_2=x_0d_0(y,z)+x_1d_1(y,z)
    =y_0c_0(x,z)+y_1c_1(x,z).
\end{equation}
The two generalized sections have one representative per chart, the
superscript naming the chart on which the display is regular:
\begin{equation}
 q_1^{(x_0)}=\frac{d_1}{x_0},\qquad
 q_1^{(x_1)}=-\frac{d_0}{x_1},
 \label{eq:x22-q1}
\end{equation}
and
\begin{equation}
 q_2^{(y_0)}=\frac{c_1}{y_0},\qquad
 q_2^{(y_1)}=-\frac{c_0}{y_1}.
 \label{eq:x22-q2}
\end{equation}
The identities $x_0d_0+x_1d_1=p_2$ and
$y_0c_0+y_1c_1=p_2$ prove equality on the intermediate locus and give the
frame transitions for $\cO(-1,1,1)$ and $\cO(1,-1,1)$.  They also show
directly that changing either rational display adds a multiple of the earlier
equation, so \ref{lem:residue-covariance} applies to both generalized
sections in the second block.

\section{Covariance of the Poincar\'e residue and cohomological data}
\label{app:globaldata}

\subsection{Gluing and covariance of the Poincar\'e residue}
\label{app:glue}

With the stage-ordered conventions of the main text, two chart regular
sequences $F_U,F_V$ for the same embedding are related by the block
lower triangular basis change of Eq.~\eqref{eq:conormal-change};
equivalently, $F_V-A_{VU}F_U$ lies in $I_X^2$.  Differentiating and
restricting to $X$ gives
\begin{equation}
 \left.\dd F_V\right|_X=\left.A_{VU}\,\dd F_U\right|_X
 \qquad\text{in }\left.\Omega_A^1\right|_X\otimes\mathbb C^c;
\label{eq:appendix-differential-transition}
\end{equation}
the terms containing $F_U\dd A_{VU}$ and the differential of the $I_X^2$
remainder vanish on $X$.  The explicit lower-triangular basis changes are
checked symbolically for the local representatives above and numerically in
every available projective chart.

\begin{lemma}[Gluing and covariance of the Poincar\'e residue]
\label{lem:residue-covariance}
Suppose $X$ is a regular embedding and
Eq.~\eqref{eq:conormal-change} holds with $A_{VU}$ invertible on $X$.
Let $z_U,z_V$ be ambient affine coordinates and write
$J_{VU}=\det(\partial z_V/\partial z_U)$.  Then the local Poincar\'e residues
obey
\begin{equation}
 \Omega_V=J_{VU}(\det A_{VU})^{-1}\Omega_U,
 \label{eq:appendix-residue-transition}
\end{equation}
including the displayed line-bundle frame factors.  In intrinsic coordinates
$u$ and $u'=u'(u)$ this becomes
\begin{equation}
 h_{u'}=h_u\det\!\left(\frac{\partial u}{\partial u'}\right),
 \qquad
 \rho_{\Omega,u'}=\rho_{\Omega,u}
 \left|\det\!\left(\frac{\partial u}{\partial u'}\right)\right|^2.
 \label{eq:appendix-rho-transition}
\end{equation}
\end{lemma}

\begin{proof}
Write $F_V=A_{VU}F_U+R$ with $R\in I_X^2$.  Both
$\dd R|_X=0$ and $F_U\dd A_{VU}|_X=0$, so
\eqref{eq:appendix-differential-transition} gives
$\dd f_{V,1}\wedge\cdots\wedge\dd f_{V,c}
=(\det A_{VU})\dd f_{U,1}\wedge\cdots\wedge\dd f_{U,c}$.  The ambient top
form contributes $J_{VU}$.  Taking the Poincar\'e residue of the complete
intersection gives
\eqref{eq:appendix-residue-transition}.  The Calabi--Yau adjunction relation
makes the combined ambient and equation-frame transition the transition of
$K_X\simeq\cO_X$.  Expressing the resulting global form as
$\Omega=h_u\dd u^1\wedge\dd u^2\wedge\dd u^3$ proves
\eqref{eq:appendix-rho-transition}.  The local formula
$h_u=\pm\det(\partial F/\partial v)^{-1}$ is the same statement after choosing
dependent coordinates $v$.
\end{proof}

Because $g_{a\bar b}$ transforms covariantly, $\det g$ has the same squared
Jacobian factor as $\rho_\Omega$.  Thus $\eta=\det g/\rho_\Omega$ and the
importance weights are invariant under projective recharting and every
admissible change of implicit coordinates with a nonsingular local Jacobian.
In the numerical coordinate selection, candidates with
$\sigma_{\min}(J_{\mathsf D})\leq10^{-12}$ are excluded and the remaining
block with the smallest $\kappa_2(J_{\mathsf D})$ is used.

The complete pointwise calculation was also repeated numerically in all
available projective and implicit-coordinate charts.  The largest relative
discrepancy among the recomputed invariant quantities is reported below.
These finite-precision checks are not used to prove covariance; they test
that the numerical realization satisfies the transition laws above.
\begin{table}[ht]
\centering
\begin{tabular}{lrrr}
\toprule
geometry & projective-chart checks & implicit-chart checks
         & maximum discrepancy\\
\midrule
$X_{11}$ & 10 & 10 & $6.68\times10^{-11}$\\
$X_{21}$ & 12 & 20 & $8.74\times10^{-12}$\\
$X_{22}$ & 24 & 35 & $1.64\times10^{-11}$\\
\bottomrule
\end{tabular}
\caption{Numerical atlas consistency checks.  The counts give the number of
independent recomputations: one per standard
affine chart of the ambient product, and one per admissible choice of
dependent coordinates in the implicit solve.}
\label{app:tab-atlas-checks}
\end{table}

\subsection{Topological data and section-space dimensions}
\label{app:rr}

Section~\ref{sec:source-immersion} requires the dimensions
\(N_k=h^0(X,L^k)\)
of the section spaces associated with the polarizations used in the
metric construction.  We obtain these dimensions from intersection
theory and Riemann--Roch.

Let \(A=\prod_\alpha \PP^{n_\alpha}\) be the ambient product and let
\(J_\alpha\) denote the hyperplane class of its \(\alpha\)th factor.
Write \(Q_a=\sum_\alpha q_a^\alpha J_\alpha\) for the class of the
\(a\)th defining section.  Restricting an integral to the zero locus of
a section multiplies the integrand by the class of that locus; applied
once per stage, and combined with the comparison of the tangent bundles
of \(X\) and \(A\), this gives, for every ambient class \(\gamma\),
\begin{equation}
 \int_X\gamma|_X=\int_A\gamma\prod_a Q_a,
 \qquad
 c(TX)=c(TA)\prod_a(1+Q_a)^{-1}\big|_X .
 \label{eq:topological-data}
\end{equation}
Here \(c(TX)\) and \(c(TA)\) are the total Chern classes of the indicated
tangent bundles.
For a negative entry in a configuration matrix, \(Q_a\) is the class of
the corresponding section on the intermediate variety where it is
defined; carried down one stage at a time, the whole calculation
reduces to polynomial arithmetic in the \(J_\alpha\), truncated by
\(J_\alpha^{n_\alpha+1}=0\).

Since \(K_X\simeq\cO_X\) and the chosen \(L\) is ample, Kodaira
vanishing gives \(H^i(X,L^k)=0\) for \(i>0\) and \(k>0\).  Threefold
Riemann--Roch therefore yields
\begin{equation}
 h^0(X,L^k)
 =
 \chi(X,L^k)
 =
 \frac{L^3}{6}k^3+
 \frac{c_2(X)\cdot L}{12}k .
 \label{eq:rr-dimension}
\end{equation}
Here \(\chi(X,L^k)\) is the holomorphic Euler characteristic, while
\(L^3\) and \(c_2(X)\!\cdot L\) are the corresponding intersection numbers.
The resulting exact dimensions are
\begin{center}
\begin{tabular}{lrrl}
\toprule
geometry & \(L^3\) & \(c_2(X)\!\cdot L\) & \(h^0(X,L^k)\)\\
\midrule
\(X_{11}\) & 23 & 86 & \((23k^3+43k)/6\)\\
\(X_{21}\) & 28 & 76 & \((28k^3+38k)/6\)\\
\(X_{22}\) & 50 & 104 & \((50k^3+52k)/6\)\\
\bottomrule
\end{tabular}
\end{center}

These formulas determine the source dimensions used in
Sec.~\ref{sec:source-immersion} and the sizes of the corresponding
unrestricted Hermitian parameter spaces.  Appendix~\ref{app:immersion} independently verifies
that the section vectors used in the computation attain these
dimensions and define the required immersions.

\section{Sampling measure and immersion of the source maps}
\label{app:certificates}

\subsection{Fibrewise sampling measure}
\label{app:crofton}

This appendix provides the geometric justification for the sampling
measure of Sec.~\ref{sec:sampling} and defines the numerical checks
reported below.  Let \(Y_\beta\subset\PP^n\) be the fibre over a base
point \(\beta\), and suppose that \(Y_\beta\) is a smooth complex
\(d_{\rm f}\)-fold of fixed projective degree.  Let
\[
 \Lambda\simeq\PP^{n-d_{\rm f}}\subset\PP^n
\]
be a complementary linear subspace drawn from the Haar measure on the
corresponding Grassmannian.  For almost every \(\Lambda\), the intersection
\(Y_\beta\cap\Lambda\) is transverse and finite.  The complex-projective
kinematic, or Crofton, formula~\cite{howard1993} states that, for any
integrable function \(\varphi\),
\begin{equation}
 \mathbb E_\Lambda\!\left[
 \sum_{p\in Y_\beta\cap\Lambda}\varphi(p)
 \right]
 =
 C\int_{Y_\beta}\varphi\,
 \frac{\omega_{\rm FS}^{d_{\rm f}}}{d_{\rm f}!}.
 \label{eq:crofton-identity}
\end{equation}
Here \(\mathbb E_\Lambda\) denotes averaging over \(\Lambda\),
\(\omega_{\rm FS}\) is the Fubini--Study form of the projective space
containing the fibre, and \(C\) depends only on the normalization of the
Haar and Fubini--Study measures.

Equation~\eqref{eq:crofton-identity} says that the complete collection of
intersection points, averaged over random \(\Lambda\), represents integration
against the Fubini--Study volume form on the fibre.  Sampling the base
point independently with the Fubini--Study measure therefore gives the
proposal forms in Eq.~\eqref{eq:proposal-forms}.  The constant \(C\) is
independent of \(\beta\) on each constant-degree family and cancels from
the normalized importance weights.  The fibres are smooth and have the
stated degree away from their algebraic discriminant locus, which has
Fubini--Study measure zero.

Random linear sections and their induced Fubini--Study measures have
previously been used to sample projective Calabi--Yau hypersurfaces and
ordinary complete intersections
\cite{douglas2006,braun2008,larfors2022}, where the point distribution
is usually justified through the random-sections result of
Ref.~\cite{shiffman1999zeros}.  The additional ingredient for
the gCICYs considered here is that the later equations are
evaluated through their regular local representatives on the appropriate
intermediate varieties, as described in Sec.~\ref{sec:representatives}.

The numbers of points in one complete fibre intersection are the
projective intersection numbers
\begin{equation}
 \int_{\PP^3}(4H)H^2=4,
 \qquad
 \int_{\PP^4}(2H)(3H)H^2=6,
 \qquad
 \int_{\PP^2}(3H)H=3.
 \label{eq:supp-fibre-degrees}
\end{equation}
Here \(H\) is the hyperplane class of the projective space containing the
fibre.  The factors \(4H\), \((2H)(3H)\), and \(3H\) describe the
quartic, \((2,3)\), and cubic fibres, while the final powers of \(H\)
represent the complementary random linear subspaces.  These three
intersection numbers give the cluster sizes \(4\), \(6\), and \(3\) in
Eq.~\eqref{eq:cluster-sizes}.  The auxiliary member
\(X_{11}^{(4)}\), used only in an independent check of the sampling
method, has the same quartic-fibre construction as \(X_{11}\).

\paragraph{Local proposal density.}

To evaluate an importance weight, the proposal form and the
Calabi--Yau volume form must be written relative to the same local
coordinate-volume element.  In intrinsic holomorphic coordinates
\(u=(u^1,u^2,u^3)\), write
\[
 \dd\nu_{\rm prop}
 =
 \rho_{{\rm prop},u}\,\dd\lambda_u,
 \qquad
 \dd\mu_\Omega
 =
 \rho_{\Omega,u}\,\dd\lambda_u.
\]
For each ambient projective factor \(\alpha\), let \(G_\alpha\) be the
\(3\times3\) Hermitian matrix obtained by pulling its Fubini--Study
metric back to the tangent space of \(X\).  Introduce formal variables
\(t_x,t_y,t_z\), used only to select the required mixed term from a
determinant expansion.

For \(X_{11}\) and \(X_{21}\), the proposal form is
\(\omega_y\wedge\omega_x^2/2\), and its local density is
\begin{equation}
 \rho_{{\rm prop},u}
 =
 [t_x^2t_y]\,
 \det(t_xG_x+t_yG_y).
 \label{eq:proposal-density-two-factor}
\end{equation}
The notation \([t_x^2t_y]\) means ``take the coefficient of
\(t_x^2t_y\)''.  For \(X_{22}\), the proposal form is
\(\omega_x\wedge\omega_y\wedge\omega_z\), and
\begin{equation}
 \rho_{{\rm prop},u}
 =
 [t_xt_yt_z]\,
 \det(t_xG_x+t_yG_y+t_zG_z).
 \label{eq:proposal-density-three-factor}
\end{equation}
These formulas follow by expanding the top exterior power of the sum of
the pulled-back Fubini--Study forms.  Under a holomorphic change of
intrinsic coordinates, \(\rho_{{\rm prop},u}\) and
\(\rho_{\Omega,u}\) acquire the same squared Jacobian factor.  Their
ratio, and hence every importance weight in
Eq.~\eqref{eq:importance-weight}, is therefore independent of the
intrinsic chart.

\paragraph{Complete-intersection rule.}

The Crofton formula involves the sum over the complete intersection
\(Y_\beta\cap\Lambda\).  We therefore accept or reject each complete intersection,
including all roots produced by one random linear section, as a unit.  An
intersection is accepted only
when

\begin{enumerate}
 \item it contains the expected number of finite roots;
 \item every root satisfies all local defining equations to the specified
 numerical tolerance; and
 \item a nonsingular intrinsic coordinate chart is available at every
 root.
\end{enumerate}

If any condition fails, the entire intersection is rejected.  All roots
of an accepted intersection are included.  Importance weights are
computed only after this decision, so neither the target measure nor the
Monge--Amp\`ere residual can influence which intersections are accepted.
The minimum projective separation between distinct roots is monitored to
detect possible numerical duplicates.  For all samples reported here it is
far above \(10^{-7}\), so imposing a separation threshold at that value
would leave the samples unchanged.

\paragraph{Numerical checks.}

Five quantities describe the completeness of the computed intersections
and the concentration of their importance weights.

The numbers \(N_{\rm try}\) and \(N_{\rm acc}\) are the numbers of
attempted and accepted complete intersections; they
count fibre intersections rather than individual points.  Each accepted
intersection contributes, respectively, \(4\), \(6\), or \(3\) points
for \(X_{11}\), \(X_{21}\), or \(X_{22}\).

For two distinct roots \(p\) and \(q\) in the same fibre, let
\(\widehat z_p\) and \(\widehat z_q\) be unit-norm homogeneous
representatives in the fibre projective space.  Their projective
separation is
\begin{equation}
 \delta(p,q)
 =
 \sqrt{1-\left|
 \widehat z_p^\dagger\widehat z_q
 \right|^2}.
 \label{eq:projective-root-separation}
\end{equation}
The quantity \(\delta_{\min}\) is the smallest value of
\(\delta(p,q)\) among distinct roots in all accepted intersections.  A
very small value would indicate nearly coincident roots or a possible
numerical duplication.

To check the defining equations, write the \(a\)-th regular local
equation as
\[
 f_a(u)=\sum_\nu c_{a\nu}m_\nu(u),
\]
where \(m_\nu\) are its monomial terms.  The relative defining-equation
residual at a computed point \(p\) is
\begin{equation}
 \varepsilon_a(p)
 =
 \frac{|f_a(p)|}
 {\max\!\left\{
 1,\sum_\nu|c_{a\nu}m_\nu(p)|
 \right\}}.
 \label{eq:relative-root-residual}
\end{equation}
The quantity \(\varepsilon_{\max}\) is the largest value of
\(\varepsilon_a(p)\) over all roots of accepted intersections and all local defining
equations.

For one numerical sample containing \(N_{\rm pt}\) points, let the
importance weights be normalized by
\[
 w_i\geq0,
 \qquad
 \sum_{i=1}^{N_{\rm pt}}w_i=1.
\]
The largest weight relative to uniform weighting is
\begin{equation}
 R_{\max}
 =
 N_{\rm pt}\max_iw_i.
\end{equation}
Uniform weights give \(R_{\max}=1\); for example,
\(R_{\max}=30\) means that the largest point weight is thirty times the
uniform value.

Order the weights as
\[
 w_{(1)}\geq w_{(2)}\geq\cdots\geq w_{(N_{\rm pt})}.
\]
The fraction of the total weight carried by the largest-weight one
percent of the points is
\begin{equation}
 W_{1\%}
 =
 \sum_{i=1}^{\lceil0.01N_{\rm pt}\rceil}w_{(i)}.
\end{equation}
Uniform weights would give \(W_{1\%}\simeq1\%\).

\begin{table}[ht]
\centering
\small
\begin{tabular}{lccccc}
\toprule
geometry
& \(N_{\rm try}/N_{\rm acc}\)
& \(\delta_{\min}\)
& \(\varepsilon_{\max}\)
& \(R_{\max}\)
& \(W_{1\%}\)\\
\midrule
\(X_{11}\)
& \(65{,}536/65{,}536\)
& \(3.41\times10^{-2}\)
& \(3.20\times10^{-15}\)
& \(33.51\)
& \(7.45\%\)\\
\(X_{21}\)
& \(8{,}192/8{,}192\)
& \(7.02\times10^{-2}\)
& \(7.44\times10^{-16}\)
& \(13.94\)
& \(4.97\%\)\\
\(X_{22}\)
& \(16{,}384/16{,}384\)
& \(5.88\times10^{-2}\)
& \(8.89\times10^{-16}\)
& \(116.51\)
& \(11.40\%\)\\
\bottomrule
\end{tabular}
\caption{Numerical checks of the complete fibre intersections and their
importance weights.  The first numerical column gives the total numbers
of attempted and accepted intersections across the independent
calculations.  The remaining columns give the smallest observed root
separation and the largest observed defining-equation residual, relative point
weight, and top-one-percent weight fraction.}
\label{app:tab-sampler-audits}
\end{table}

Every attempted intersection in these checks was accepted with the
expected number of roots.  The defining-equation residuals are close to
floating-point precision, and no nearly coincident roots were observed in these finite
samples.  The importance weights are nonuniform, most visibly for
\(X_{22}\): its largest observed point weight is \(116.51\) times the
uniform value, and the largest-weight one percent of the points carry
\(11.40\%\) of the total weight.  These quantities enter the effective
sample sizes reported in Sec.~\ref{sec:training}.

These calculations test numerical completeness and weight concentration on
the sampled fibres; they are not interval-arithmetic proofs for every
floating-point root.

\subsection{Source-space completeness and immersion}
\label{app:immersion}

The source vectors used in the metric calculation consist of restrictions
of ambient monomials.  Denote their source line bundle by \(L_{\rm src}\).
For \(X_{11}\), the auxiliary presentation has degree \((2,2,0)\)
includes the auxiliary \(\PP^1\) introduced by the polynomial
stabilization.  If \([z_0:z_1]\) are its homogeneous coordinates, the
equation \(z_0+z_1=0\) fixes that factor, so its effective
source degree is \((2,2)\).

\begin{table}[ht]
\centering
\begin{tabular}{lrrrr}
\toprule
geometry & effective source degree &
\(h^0(X,L_{\rm src})\) & prime \(p\) &
nonzero minor modulo \(p\)\\
\midrule
\(X_{11}\) & \((2,2)\)     & 45 & 11 & 2\\
\(X_{21}\) & \((1,1)\)     & 11 & 11 & 4\\
\(X_{22}\) & \((1,1,1)\)   & 17 & 13 & 8\\
\bottomrule
\end{tabular}
\caption{Exact restriction-rank witnesses for the source vectors used in
the metric calculation.}
\label{app:tab-source-immersion}
\end{table}

At each indicated prime, exact reduction gives a three-dimensional
complete intersection and a \(p\)-torsion-free integral model.
Evaluating the chosen sections at exact finite-field points produces
the nonzero maximal minors displayed in
Table~\ref{app:tab-source-immersion}.  Their nonvanishing proves that the
chosen sections are linearly independent in characteristic zero.
Appendix~\ref{app:rr} gives the matching dimensions
\[
h^0(X,L_{\rm src})=45,\ 11,\ 17,
\]
so the chosen source vectors are bases of the complete restricted
section spaces.

The corresponding ambient line bundles are very ample and define Segre
or Segre--Veronese embeddings.  Restricting these embeddings to \(X\), and
removing only section relations that vanish identically on \(X\), preserves
a closed immersion.  An invertible reference Hermitian form changes only
the projective coordinates, while its \(m\)-fold power composes the map
with a Veronese embedding.  Both operations preserve immersion.
Consequently, the source maps used in Proposition~\ref{prop:positive} are
immersions for all three geometries.

\section{Numerical settings}
\label{app:settings}

This appendix records the initialization, optimization, validation and sample
settings for the calculations in Sec.~\ref{sec:results}.

\subsection{Reference-term scale and function-preserving expansions}
\label{app:eps-and-activation}

The coefficient \(\varepsilon_{\rm ref}\) in
Eq.~\eqref{eq:positive-density} is a numerical setting, not a universal
constant.  We use \(10^{-4}\) for the \(X_{11}\) calculations, the
matched \(X_{21}\) bond comparison and the \(X_{22}\) calculations.
The exploratory \(X_{21}\) degree--bond comparison uses \(10^{-10}\).
Every direct comparison uses a common value of \(\varepsilon_{\rm ref}\).
Write \(F_{B,\varepsilon}=\|Bv\|^2+\varepsilon F_0^m\) for the density with
factor \(B\) and reference coefficient \(\varepsilon\).  For any two positive
choices \(\varepsilon_1\) and \(\varepsilon_2\), the change can be absorbed by
the unrestricted overall scale of the factor:
\begin{equation}
 B_2=\sqrt{\frac{\varepsilon_2}{\varepsilon_1}}\,B_1
 \quad\Longrightarrow\quad
 F_{B_2,\varepsilon_2}
 =\frac{\varepsilon_2}{\varepsilon_1}F_{B_1,\varepsilon_1}.
 \label{eq:epsilon-scale-equivalence}
\end{equation}
The corresponding potentials differ only by an additive constant and
therefore define the same metric.  Thus \(\varepsilon_{\rm ref}\) changes
the parameter scale, initialization and numerical conditioning, but not the
family of metrics represented by an unconstrained \(B_\theta\).

Bond enlargements preserve the represented function as follows.  For one expanded bond,
suppressing physical indices and the other bond axes, the initialization
has the block form
\begin{equation}
 \widetilde C^{[j]}=
 \begin{pmatrix}C^{[j]}&U^{[j]}\end{pmatrix},
 \qquad
 \widetilde C^{[j+1]}=
 \begin{pmatrix}C^{[j+1]}\\0\end{pmatrix}.
\label{eq:function-preserving-bond-expansion}
\end{equation}
The new path contributes \(U^{[j]}0=0\), so the represented function is
unchanged.  The derivative with respect to the zero incoming row of the
next core is proportional to the nonzero \(U^{[j]}\), however, and is
generically nonzero.  At the exact embedding point, the derivative with
respect to \(U^{[j]}\) itself is still zero because it is multiplied by the
zero row in the next core.  Activation is therefore sequential: the new
incoming row moves first, after which gradients can also rotate
\(U^{[j]}\).  Applying this embedding at every internal bond preserves the
function before joint optimization.  Initializing
both sides at zero would instead leave the entire new path at a stationary
zero contribution.

The same construction applies to dictionary enlargements.  A new
local operator \(Q_\alpha\) multiplies its coefficient-core slice, so
initializing both factors at zero would leave the pair at a stationary
point.  Instead the new \(Q_\alpha\) are added as a nonzero orthonormal
complement of the existing dictionary while their coefficient slices
start at zero: the represented function is unchanged, the derivative
with respect to each new slice is generically nonzero, and subsequent
joint training can also rotate the dictionary.

\subsection{Pointwise normalization and contraction cost}
\label{app:pointwise-evaluation}

To prevent overflow during the sitewise contraction, the section values and
their first derivatives are divided at each point by the same positive scalar
obtained from the reference norm.  By homogeneity this multiplies the
representatives of \(F\),
\(\partial F\) and \(\partial\bar\partial F\) by one common factor; the
corresponding logarithmic correction is recorded separately.  This is a
rescaling of the stored jet, not a change of the underlying function.  With
\(F_a=\partial_aF\), \(F_{\bar b}=\bar\partial_{\bar b}F\), and
\(F_{a\bar b}=\partial_a\bar\partial_{\bar b}F\), the common factor cancels from
\begin{equation}
 \partial_a\bar\partial_{\bar b}\log F
 =\frac{F_{a\bar b}}{F}-\frac{F_aF_{\bar b}}{F^2},
\end{equation}
and restoring the recorded logarithmic correction preserves \(\log F\).
The procedure therefore preserves the first and mixed logarithmic
derivatives in exact arithmetic.

The contraction cost was measured for \(X_{11}\) in complex128
on an RTX 4090, using batch size \(256\), \(d=45\) and \(D=5\).  A fit to
the median forward time over \(m=3,\ldots,8\) has slope \(2.15\) ms per
site and coefficient of determination \(R^2=0.9999\); the two smallest cases
have no internal \(D^2\)
transfer block and lie off this line.  For \(m=3\), the
model gives \(3.9\) ms for the metric forward pass and \(13.3\) ms for a
complete optimizer update.  These timings characterize the TN contraction;
they are not a speed comparison with the neural \(\phi\)-model.

\subsection{Settings for \texorpdfstring{\(X_{11}\)}{X11}}
\label{app:x11-settings}

\paragraph{TN--\(\phi\) comparison.}
Both models use the same \(196{,}608\) training points, \(24{,}576\)
validation points, geometric loss coefficients, and batch size \(256\).
Before each update, a total gradient norm above \(5\) is rescaled to \(5\).
Both begin with a 20-epoch initial stage.  For
the subsequent primary and precision stages, the \(\phi\)-model selects rates from
\((3\times10^{-4},10^{-4})\), \((10^{-3},3\times10^{-4})\), and
\((3\times10^{-3},10^{-3})\); validation selects the middle pair.  The TN is
not swept and uses \(2\times10^{-4}\) and \(5\times10^{-5}\).  In each stage,
eight validation evaluations without a relative gain of \(0.5\%\) (primary)
or \(0.3\%\) (precision) stop training.  The validation sample alone controls
selection and stopping; a separate \(49{,}152\)-point sample is diagnostic
only.

The TN uses complex128 and the \(\phi\)-model float64.  Denoting the fully
connected neural network with parameters \(\vartheta\) by \(N_\vartheta\), the
neural correction is \(\phi_\vartheta=0.1N_\vartheta(\rho^{(s_2)})\), with the prefactor fixing
its initial scale.  Only when \(e^{\widetilde\ell_i}\) is formed inside the
training loss, \(\widetilde\ell_i\) is restricted to \([-20,20]\) for the TN
and unrestricted \(H_6\), and to \([-15,10]\) for the \(\phi\)-model.  This
prevents numerical overflow; the logarithmic loss and all reported
observables use the unrestricted value.

\paragraph{Optimization of unrestricted \(H_6\).}
The three \(H_6\) paths use the same training and validation points and batch
size \(512\); total gradient norms above \(2\) are rescaled to \(2\).  The
direct cubic lift was checked on an independent \(2{,}048\)-point sample;
the maximum relative errors in the product-section relation and metric
entries are \(2.8\times10^{-14}\) and \(1.2\times10^{-9}\), respectively.
Starting from this lift, optimization proceeds in
blocks of \(30{,}000\) updates with a fixed logarithmic learning-rate
sequence from \(3\times10^{-5}\) to \(10^{-9}\), retaining the lowest
validation value between blocks.  It continues for at least \(1200\) s,
longer than the slowest TN optimization (\(1133\) s), and ends with a
validation-based learning-rate reduction from \(3\times10^{-6}\).  The
three optimization times are \(1502\), \(1386\), and \(1331\) s.  These
times exclude the one-time construction of section jets.

\subsection{Settings for \texorpdfstring{\(X_{21}\)}{X21} and \texorpdfstring{\(X_{22}\)}{X22}}
\label{app:x21-x22-settings}

\paragraph{Initial \(X_{21}\) continuation.}
The three \(D=8\) source models and three \(D=12\) continuations share one
data division.  Equation~\eqref{eq:function-preserving-bond-expansion}
preserves the function at the \(D=8\to12\) embedding, after which all cores are
trained jointly.  The complex128 \(D=12\) runs use Adam for 30 epochs, batch
size \(1024\), learning rate \(2\times10^{-4}\), and gradient-rescaling
threshold \(2\).

\paragraph{Matched \(X_{21}\) bond-dimension comparison.}
Each fixed \(D=8\) model is continued once at \(D=8\) and once after a
function-preserving expansion to \(D=12\).  Within each pair, the 30 epochs,
learning rate, batch size, gradient-rescaling threshold, minibatch sequence, and \(5{,}760\)
updates are identical; only the additional bond channels differ.

\paragraph{\(X_{22}\) dictionary and degree comparisons.}
The \(q=22\), \(q=60\), and fixed \(q=289\) degree-four models share the same
selected \(q=22\) starting metric; the two enlargements preserve it to
numerical precision.  They share \(131{,}064\) training points, \(32{,}766\)
validation points, minibatch order, and complex128 Adam settings: 20 epochs,
batch size \(64\), learning rate \(5\times10^{-5}\), and gradient-rescaling
threshold \(2\).
The \(q=22\) and \(q=60\) dictionaries are trainable; the complete \(q=289\)
dictionary is fixed and only its coefficient cores are trained.  The
\(q=60,k=6\) continuation uses the same optimizer settings with independently
generated samples.
Each reported model is selected by its validation score subject to sampled
positivity.

\bibliographystyle{unsrtnat}
\bibliography{references}

\end{document}